\let\CoherenceKernelLabel\label
\documentclass[aps,pra,10pt,onecolumn,superscriptaddress,nofootinbib]{revtex4-2}

\usepackage{amsmath,amssymb,amsthm,mathtools}
\usepackage{bm}
\usepackage{graphicx}
\usepackage{xcolor}
\usepackage{booktabs}
\usepackage{enumitem}
\usepackage{tikz}
\usepackage{braket}

\DeclareMathOperator{\Tr}{Tr}

\DeclareMathOperator{\Span}{span}
\DeclareMathOperator{\diag}{diag}

\DeclareMathOperator{\rank}{rank}

\makeatletter
\@ifl@t@r\fmtversion{2023-06-01}{%
  \let\label\CoherenceKernelLabel
  \let\ltx@label\label
}{}
\makeatother
\usepackage{hyperref}
\hypersetup{colorlinks=true, linkcolor=blue, citecolor=blue, urlcolor=blue}

\newcommand{\TO}{\ensuremath{\mathrm{TO}}}

\newcommand{\CGP}{\ensuremath{\mathrm{CGP}}}
\newcommand{\TOfin}{\ensuremath{\mathrm{TO}_{\mathrm{fin}}}}
\newcommand{\cTO}{\ensuremath{\overline{\mathrm{TO}_{\mathrm{fin}}}}}
\newcommand{\ee}[2]{\ket{#1}\!\bra{#2}}
\newcommand{\one}{\mathbf{1}}

\allowdisplaybreaks[2]

\newtheorem{theorem}{Theorem}
\newtheorem{lemma}{Lemma}

\newtheorem{corollary}{Corollary}
\newtheorem{proposition}{Proposition}

\begin{document}

\title{Limitations on joint coherence preservation in quantum thermodynamics}

\author{Piotr {\'C}wikli{\'n}ski}
\affiliation{International Centre for Theory of Quantum Technologies, University of Gda{\'n}sk, ul. prof. Marii Janion 7, 80-309 Gda\'nsk, Poland}

\author{Micha{\l} Studzi{\'n}ski}
\affiliation{International Centre for Theory of Quantum Technologies, University of Gda{\'n}sk, ul. prof. Marii Janion 7, 80-309 Gda\'nsk, Poland}

\date{\today}

\begin{abstract}
Coherences between different energy levels are strongly constrained
by thermodynamics. Here we ask whether several coherence transfers
that can be optimized separately can also be optimized simultaneously
by the same thermodynamic process. Using an energy-resolved
Stinespring representation of thermal operations, we express
elementary transfer amplitudes as inner products of bath-weighted
vectors. Joint saturation imposes collinearity and phase-closure
conditions on these vectors, which must all arise from one
energy-preserving system--bath unitary. This gives necessary
compatibility conditions for arbitrary finite-dimensional
Hamiltonians. We show that joint incompatibility persists in a
system with both degenerate and nondegenerate energy levels:
rigorous computer-assisted semidefinite-program bounds, combined
with explicit thermal constructions, prove that three individual
thermal optima cannot be approached simultaneously. We then
determine how well different coherences can be protected when
their optima are incompatible. For a four-level family at infinite
temperature, we obtain the complete closed retention region and
the cost of balancing the coherences relative to their individual
optima. At positive temperature, covariance and Gibbs preservation
permit joint retention that no thermal operation can approach.
Finite-bath constructions and certified semidefinite relaxations
bound the balanced thermal performance. When two coherences carry
independent phases, these bounds describe simultaneous preservation
of interference visibility and phase Fisher information.
\end{abstract}

\maketitle

\section{Introduction}
\label{sec:intro}

Suppose that a prescribed change of populations can be carried
out while preserving a given coherence as well as possible.
We may ask the same question for another coherence, and obtain
another optimal process. Can one thermal operation preserve both
coherences at their individual optima? More generally, conditions
on each coherence transfer separately need not determine whether
several such transfers can occur together. The question is one
of compatibility, since all amplitudes must be produced by the
same physical process.

Thermodynamic restrictions on quantum state transformations
concern both the populations and the coherences of the state
\cite{Streltsov2017,HorodeckiOppenheim2013,Brandao2015,
Lostaglio2019RoPP,ChitambarGour2019,Gour2015,Gour2022RoleCoherence,
Marvian2022QFI,Kondra2024CoherenceManipulation,KorzekwaWorkFromCoh2016}.
Under a thermal operation, one couples the system to a bath
in a Gibbs state, applies a unitary commuting with the total
Hamiltonian, and then discards the bath
\cite{Janzing2000,HorodeckiOppenheim2013}.
The resulting channel is covariant under time translations
and preserves the system's Gibbs state. These properties restrict
how coherences between different energy levels evolve
\cite{Lostaglio2015NatComm,Lostaglio2015PRX,
PhysRevLett.115.210403}.
Time-translation covariance relates these restrictions to
the resource theory of asymmetry
\cite{MarvianSpekkens2014,MarvianSpekkensModes2014,
MarvianSpekkens2016Coherence}.

For states diagonal in the energy basis, thermomajorization
characterizes which state-to-state transformations can be
approached arbitrarily closely using finite-bath thermal
operations \cite{HorodeckiOppenheim2013,Renes2014,Egloff2015,
ScharlauMuller2018}. Population transformations have also been
studied with finite tolerance, with degeneracies, and under
Markovian dynamics
\cite{ApproxMajorization2018,VomEndeMalvetti2024ThermomajorizationPolytope,
LostaglioKorzekwa2022ContinuousThermomajorization,
KorzekwaLostaglio2022OptimizingThermalization}.
These conditions do not yet tell us how much coherence can
survive the transformation. When each nonzero Bohr mode is
one dimensional, covariance prevents different coherences
from mixing, and damping-matrix positivity bounds their
retention by
\[
|\kappa_{ij}|^2\le P_{ii}P_{jj}.
\]
Here \(\Phi(\ee ij)=\kappa_{ij}\ee ij\), and \(P_{ii}\)
is the probability of staying in level \(i\).
For qubits the optimal retention can be approached by thermal
operations; already for qutrits, exact saturation can fail
\cite{PhysRevLett.115.210403,HuDing2019}.
Related restrictions occur at low temperature
\cite{NarasimhacharGour2015}.
For prescribed classical transition matrices, coherence
preservation has been studied for covariant channels and through
the Jamio{\l}kowski state of a channel
\cite{LostaglioKorzekwaMilne2017,KorzekwaEtAl2018Coherifying}.

The population constraint leaves a further choice. Fixing the
final populations of one input state does not usually fix the
whole transition matrix \(P\). Different admissible matrices
can favor different coherences. The convex structure and
decomposability of thermal processes are relevant to this choice
\cite{Mazurek_2018,KorzekwaL2020-TO-advantage}.
There is also a distinction between thermal operations and
the larger class of covariant Gibbs-preserving (CGP) maps
\cite{Faist2015,Gour2018QuantumMajorization,Ding2021GapTOEnTO}.
For a thermal operation, all transition amplitudes must arise
from one energy-preserving unitary on system and bath.
This imposes restrictions beyond the channel-level scalar bounds,
with the bath spectrum and its degeneracies entering the
realization \cite{ScharlauMuller2018,VomEnde2022BathHamiltonians}.

In this work, we first formulate the compatibility conditions
without restricting the degeneracies of the Hamiltonian or its
Bohr frequencies. An energy-resolved representation of the
thermal unitary relates coherence transfer to overlaps of bath
amplitudes. The conditions for saturating these overlaps must
hold together, and must respect trace and Gibbs preservation.
We then consider a Hamiltonian with both degenerate and
nondegenerate energy levels. Using explicit thermal operations
and a rigorous semidefinite-program bound, we show that the
additional freedom within degenerate eigenspaces does not
guarantee compatibility of the individual optima.

Once joint optimality fails, the next question is how well the
coherences can be protected together. We compare maximizing
their sum with maximizing the smallest fraction of the
individual optima. For a four-level family at infinite
temperature, we determine the complete closed region of two
retained coherences. At positive temperature, we show that
covariance and Gibbs preservation permit joint retention that
no thermal operation can approach, and bound the best balanced
thermal performance using finite-bath constructions and
certified semidefinite relaxations. When the two coherences
carry independent phases, these results also describe how well
both interference signals survive the same population change.

The processes studied here use Gibbs baths without an additional
coherent resource or catalyst. Elementary and implementable
thermal operations provide useful constructions
\cite{LostaglioAlhambraPerry2018ElementaryTO,
PerryCwiklinskiAndersHorodeckiOppenheim2018ImplementableTO,HackMendl2025}.
Catalytic assistance and additional coherent resources define
further tasks
\cite{Aberg2014,PhysRevResearch.6.033203,SonNg2025HierarchyCatalysis,
LipkaBartosikWilmingNg2024CatalysisReview,PhysRevLett.134.160402,
TajimaTakagi2025}.
The general channel reductions, proofs, and bath constructions
are collected in the appendices. Further supporting proofs and
the fixed verification data accompany the paper in the
Supplemental Material~\cite{SM}.

\section{General constraints and joint incompatibility}
\label{sec:framework}

\subsection{Thermal operations and conditions for joint saturation}

Let \(H_S=\sum_\alpha E_\alpha\Pi_\alpha\) be a
finite-dimensional Hamiltonian, allowing degenerate eigenspaces.
A finite-bath thermal operation at inverse temperature \(\beta\)
has the form
\begin{equation}
\Phi(\rho)=\Tr_B[U(\rho\otimes\tau_B)U^\dagger],
\qquad
\tau_B=\frac{e^{-\beta H_B}}{\Tr e^{-\beta H_B}},
\qquad [U,H_S+H_B]=0 .
\label{eq:TOdef}
\end{equation}
We write \(\TO=\TOfin\) for this class and \(\cTO\) for
its channel-norm closure. Finite-bath suprema need not be
attained by a finite bath. Every thermal operation is covariant
under time translations and preserves
\(\gamma=e^{-\beta H_S}/\Tr e^{-\beta H_S}\), so
\[
\TO\subseteq\cTO\subseteq\CGP.
\]
No Markovian assumption is imposed.

In the nondegenerate examples, we prescribe the complete output
diagonal \(p'\) for an input with diagonal \(p\). With
\(P_{ri}=\bra r\Phi(\ee ii)\ket r\), this constraint is
\(Pp=p'\). We denote the feasible set in channel class
\(\mathfrak F\) by \(\Gamma_{\mathfrak F}(\rho,p')\).
The individual and joint optimizations always use
the same channel class and the same target \(p'\).
In the degenerate example, only the total occupation of each
energy eigenspace is fixed; we write this feasible set as
\(\Gamma(\rho,p^f;\mathfrak F)\). The full definitions of
both sets are given in Appendix~\ref{app:general}.

We first examine the amplitude constraints without assuming
that energy levels or Bohr frequencies are nondegenerate.
Write \(\tau_B=\sum_\mu q_\mu\ee{\mu}{\mu}\), where
\(\epsilon_\mu\) is the energy of the bath state
\(\ket\mu\). For each system transition \(i\to r\),
collect the corresponding matrix elements of the thermal
unitary in the vector
\begin{equation}
w_{r|i}=\sum_{\mu,\nu}\sqrt{q_\mu}\,
U_{r\nu,i\mu}\ket{\mu,\nu}.
\label{eq:previewShellVector}
\end{equation}
Energy conservation allows a component only when
\(E_i+\epsilon_\mu=E_r+\epsilon_\nu\). These vectors
belong to one auxiliary space and are all obtained from the
same system--bath unitary. Direct expansion gives
\begin{equation}
P_{ri}=\|w_{r|i}\|^2,\qquad
\Phi(\ee ij)_{rs}=\langle w_{s|j},w_{r|i}\rangle,
\qquad
|\Phi(\ee ij)_{rs}|\le\sqrt{P_{ri}P_{sj}}.
\label{eq:previewAmplitudeBound}
\end{equation}
The transfer vanishes unless \(E_i-E_j=E_r-E_s\).
When the right-hand side of the bound is nonzero, equality
holds precisely when the two amplitude vectors are collinear.

Several input coherences can contribute to the same output
coherence if they belong to the same Bohr mode. In this case
their phases matter as well:
\begin{equation}
|\Phi(\rho)_{rs}|
\le\sum_{E_i-E_j=E_r-E_s}
|\rho_{ij}|\sqrt{P_{ri}P_{sj}}.
\label{eq:previewModeBound}
\end{equation}
Equality requires saturation of every contributing nonzero
elementary bound and a common phase for the resulting complex
terms. If several selected transfers require successive
collinearity relations, all nonzero vectors connected by those
relations must lie in the same one-dimensional complex subspace.
The resulting pattern must also be consistent with the
orthogonality of the rows and columns of the thermal unitary.
Thus checking each overlap separately is not enough to establish
joint saturation. The full representation and its equality
conditions are derived in Appendix~\ref{app:general}.

Trace and Gibbs preservation give further relations among
the same amplitudes. Writing \(\gamma_i=\bra i\gamma\ket i\),
we have
\begin{equation}
\sum_r\langle w_{r|j},w_{r|i}\rangle=\delta_{ij},
\qquad
\sum_i\gamma_i\langle w_{s|i},w_{r|i}\rangle
=\gamma_r\delta_{rs}.
\label{eq:previewGramClosure}
\end{equation}
These identities also hold for the Kraus vectors of a trace-
and Gibbs-preserving channel. For example, take distinct
equal-energy labels \(i,j\), and put
\(z_r=\Phi(\ee ij)_{rr}\) and
\(a_r=\sqrt{P_{ri}P_{rj}}\). Trace preservation requires
\(\sum_r z_r=0\). Hence simultaneous saturation
\(|z_r|=a_r\) for all \(r\) requires phases for which
the complex terms close, and in particular
\begin{equation}
2\max_r a_r\le\sum_r a_r.
\label{eq:previewPolygon}
\end{equation}
There is an analogous Gibbs-weighted condition for
\(\gamma_i\Phi(\ee ii)_{rs}\) with \(r\ne s\).
Under covariance, these trace and Gibbs applications are
nontrivial only for equal-energy input and output labels,
respectively. The conservation identities apply generally,
but the nontrivial use of these particular polygon tests
concerns zero-frequency coherences.

A failed polygon test can already make an elementary bound
unattainable for one transfer. It does not by itself establish
incompatibility of the actual individual optima. To formulate
that question, let \(q_k(\Phi)\) be the selected coherence
scores, and let \(\mathcal F\) be a common feasible set
with the channel class and population constraint fixed. Put
\[
M_k=\sup_{\Phi\in\mathcal F}q_k(\Phi),\qquad
S=\sup_{\Phi\in\mathcal F}\sum_kq_k(\Phi),\qquad
\Delta=\sum_kM_k-S.
\]
A positive \(\Delta\) rules out even a sequence of
feasible operations approaching all the individual suprema
simultaneously. One way to prove this is to find a feasible
thermal operation \(W_k\) for each score and an upper
bound \(U\) on their sum over a class containing
\(\mathcal F\). Then
\[
\Delta\ge\sum_kq_k(W_k)-U.
\]
The general statement is proved in
Lemma~\ref{lem:gapWitness}.

There are two sources of competition between the optima.
Different coherences can favor different population matrices
even when those matrices give the same output populations.
In addition, the amplitudes required at a fixed matrix may
fail to come from one thermal unitary. For simple nonzero
Bohr modes, one CGP channel saturates all scalar retention
bounds at every fixed \(P\). In that case the CGP joint
gap is entirely due to the choice of population matrix.
Thermal operations can impose the further restriction of a
common energy-preserving dilation. Appendix~\ref{app:general}
gives the fixed-\(P\) construction and separates these
two contributions.

\subsection{Joint incompatibility in the presence of degeneracy}
\label{sec:cgp-certificate}

Degenerate energy levels allow several amplitudes within a
mode and additional freedom in choosing the channel. Does this
freedom guarantee that coherences which can be optimized
separately can also be optimized together? We answer this
question by considering \(H=\diag(0,1,1,3)\) in the basis
\((\ket0,\ket{1a},\ket{1b},\ket3)\), with \(\beta=1\), so
\begin{equation}
\gamma=\frac{\diag(1,e^{-1},e^{-1},e^{-3})}
{1+2e^{-1}+e^{-3}}.
\label{eq:certificateGibbs}
\end{equation}
The input and target occupations are
\begin{equation}
\rho=
\begin{pmatrix}
0.35&0.05&0&0.04\\
0.05&0.25&0&0\\
0&0&0.20&0.05\\
0.04&0&0.05&0.20
\end{pmatrix},
\qquad
p^f=(0.50,0.35,0.15).
\label{eq:certificateState}
\end{equation}
The middle target fixes only
\(\sigma_{1a,1a}+\sigma_{1b,1b}=0.35\). The split within
that eigenspace and the full transition matrix remain
optimization variables. We select
\begin{equation}
L_1=\operatorname{Re}\sigma_{0,1a},\qquad
L_2=\operatorname{Re}\sigma_{1b,3},\qquad
L_3=\operatorname{Re}\sigma_{0,3}.
\label{eq:mixedObjectives}
\end{equation}
Their graph is the path \(1a-0-3-1b\). Energy-preserving
diagonal unitaries can align these three phases, so the same
individual and summed optima result from their magnitudes;
see Appendix~\ref{app:general}.

For this fixed input and population task, the directional
objectives are linear in the Choi matrix. Covariance splits
that matrix into energy-shift blocks, complete positivity
requires the blocks to be positive, and trace preservation,
Gibbs preservation, and the target occupations give linear
constraints. The resulting optimization over CGP maps is
therefore a semidefinite program. Certified primal and dual
bounds are obtained by exact equality corrections and interval
positivity checks; Appendix~\ref{app:mixed} gives the
formulation and certificates. To infer a gap for thermal
operations themselves, we combine the joint CGP upper bound
with separate feasible thermal constructions.

Three explicit thermal operations, one for each selected
objective, satisfy the same target occupations and give
\[
L_1(W_{7/8})+L_2(\widehat W)+L_3(W_{1/7})>0.09634328.
\]
In contrast, a certified Choi bound gives
\(\sup_{\CGP}(L_1+L_2+L_3)<0.088126375461\) for this
population task. Since every thermal operation is CGP,
these two inequalities prove a strict thermal gap without
requiring the individual thermal suprema to be known.

\begin{theorem}[Incompatibility of the individual thermal suprema]
\label{thm:mixedTOGap}
For the mixed-degenerate task, let \(\mathfrak F=\TO\) or
\(\mathfrak F=\cTO\). Then
\begin{equation}
\sum_{k=1}^3\sup_{\Gamma(\rho,p^f;\mathfrak F)}L_k
-
\sup_{\Gamma(\rho,p^f;\mathfrak F)}\sum_{k=1}^3L_k
>0.00821690 .
\label{eq:mixedTOGap}
\end{equation}
The same bound holds for the three magnitude objectives.
\end{theorem}

Indeed, each separate thermal supremum is at least its
witness score, while the joint thermal score is bounded
by the CGP value. Their difference is larger than
\(0.09634328-0.088126375461=0.008216904539\).
The bath dimensions of the three witnesses are 125, 25,
and 125. Appendix~\ref{app:mixed} gives their collision
sequences and the Choi certificate.
Thus mixed degeneracy does not guarantee joint optimality,
even when only energy-eigenspace occupations are prescribed.
This conclusion does not require a violation of the polygon
condition in Eq.~\eqref{eq:previewPolygon}.

\section{Optimal protection of several coherences}
\label{sec:protection}
\label{sec:phase}

The failure of joint optimality leaves a quantitative question:
how much of each coherence can survive one prescribed
population transformation? Maximizing the sum can favor one
coherence over another. If both carry information that must
be preserved, the smaller fraction of the individual optima
is also relevant. We make this distinction explicit using
two coherences that encode independent phases.

Consider the two-phase family
\begin{equation}
\rho(\theta_{\mathrm e},\theta_{\mathrm o})
=
\begin{pmatrix}
p_0&0&0&c_{\mathrm e}e^{-i\theta_{\mathrm e}}\\
0&p_1&c_{\mathrm o}e^{-i\theta_{\mathrm o}}&0\\
0&c_{\mathrm o}e^{i\theta_{\mathrm o}}&p_2&0\\
c_{\mathrm e}e^{i\theta_{\mathrm e}}&0&0&p_3
\end{pmatrix},
\qquad
c_{\mathrm e}^2<p_0p_3,\quad
c_{\mathrm o}^2<p_1p_2 .
\label{eq:phaseFamily}
\end{equation}
The phases occur in two separate \(2\times2\) blocks, and
one fixed time-translation covariant channel acts on the whole
family independently of the unknown
\(\theta_{\mathrm e},\theta_{\mathrm o}\).
For a nondegenerate Hamiltonian with one-dimensional nonzero
Bohr modes, the channel multiplies each selected coherence by
a parameter-independent coefficient, while its output
populations do not depend on either phase.

Let \(z_\ell\) be the output coherence in block \(\ell\)
and \(s_\ell>0\) its total population. Conditioned on
occupying this block, a relative phase shift \(\varphi\)
followed by a balanced interference measurement gives
\begin{equation}
p_+(\varphi)=\frac12+
\frac{\operatorname{Re}(e^{i\varphi}z_\ell)}{s_\ell}.
\label{eq:previewInterference}
\end{equation}
The fringe visibility and the block's contribution to the
quantum Fisher information of the full state are
\begin{align}
V_\ell&=\frac{2|z_\ell|}{s_\ell},
\label{eq:previewVisibility}\\
\mathcal F_{\ell\ell}&=\frac{4|z_\ell|^2}{s_\ell},\qquad
\mathcal F_{\mathrm e\mathrm o}=0.
\label{eq:previewFisher}
\end{align}
A block of zero weight contributes zero Fisher information.
A measurement which first resolves the two blocks and then
reads the corresponding phase attains this matrix locally
at a calibrated phase point; see Appendix~\ref{app:phase}.
Since the output populations are fixed, the denominators in
these expressions are fixed as well. Preserving a coherence
therefore preserves both its visibility and its phase information.
Coherent control may be used for preparation and readout, but
no additional coherent resource is supplied during the thermal
operation.

For these two coherences, we take \(q_k(\Phi)=|z_k(\Phi)|\).
With \(M_k>0\), balanced protection maximizes
\[
\alpha_{\mathcal F}
=\sup_{\Phi\in\mathcal F}\min_k\frac{q_k(\Phi)}{M_k}.
\]
For one channel, write \(\alpha(\Phi)=\min_kq_k(\Phi)/M_k\).
The normalizers are the actual individual suprema over the
same feasible set. Thus \(\alpha\) compares each coherence
with what could be achieved if that coherence alone mattered.
For the phase family at fixed output populations, it is also
the smallest fraction of the individual visibility optima;
the corresponding fraction of phase Fisher information is
\(\alpha^2\).

\subsection{An exact tradeoff at infinite temperature}
\label{sec:analytic}

We begin with a family for which the full tradeoff can be
found analytically. Take four increasing energies with six
distinct positive gaps, and set \(\beta=0\). Choose
\begin{equation}
0<h<\frac16,\qquad a=\frac{1-6h}{4},\qquad
p=a\one+h(3,2,1,0)^T,
\qquad
p'=p+hx(-1,-1,1,1)^T,
\quad 0<x<1.
\label{eq:analyticPopulations}
\end{equation}
Here \(h\) is a population spacing. Let
\begin{equation}
\rho=\diag p+c(\ee03+\ee30+\ee12+\ee21),
\qquad 0<c<\sqrt{a(a+3h)} .
\label{eq:analyticInput}
\end{equation}
This state is positive definite and normalized; both coherent
blocks have population \(1/2\), before and after the prescribed
change. Write the normalized retentions as
\[
r_{\mathrm e}=\frac{|\Phi(\rho)_{03}|}{c},\qquad
r_{\mathrm o}=\frac{|\Phi(\rho)_{12}|}{c},
\]
and let \(\mathcal R_{\mathfrak F}\) be their attainable
region for \(\mathfrak F\in\{\TO,\cTO,\CGP\}\), with
the full output diagonal fixed to \(p'\).

\begin{theorem}[Complete selected-retention region]
\label{thm:analyticRegion}
For the task in Eqs.~\eqref{eq:analyticPopulations}
and~\eqref{eq:analyticInput},
\begin{equation}
\mathcal R_{\cTO}=\mathcal R_{\CGP}
=
\left\{(r_{\mathrm e},r_{\mathrm o})\ge0:
\begin{aligned}
r_{\mathrm e}&\le1-x/3,\\
r_{\mathrm o}&\le1-x/2,\\
3r_{\mathrm e}+r_{\mathrm o}&\le4-2x
\end{aligned}\right\}.
\label{eq:analyticRegion}
\end{equation}
The Pareto boundary joins
\[
A_x=(1-x/2,1-x/2),\qquad
B_x=(1-x/3,1-x).
\]
Every point is a limit of finite-bath TO with the same full
output diagonal \(p'\) imposed exactly at each step.
Every nonnegative point satisfying the three upper
inequalities strictly is attainable by a finite bath.
\end{theorem}

The two endpoints correspond to different population flows.
The route attaining \(A_x\) exchanges populations in pairs
\(02\) and \(13\); the route attaining \(B_x\) uses
\(03\) and \(12\). Both give the same final populations.
Uniform Gibbs ladders implement these routes with exactly
the prescribed transition matrices and coherence errors
which vanish as the bath grows. Mixtures and independent
attenuation of the two coherences give the rest of the region.
The proof and finite-bath construction are in
Appendix~\ref{app:infinite}.

The individual normalized maxima are
\(m_{\mathrm e}=1-x/3\) and \(m_{\mathrm o}=1-x/2\).
They cannot be attained together: their pair lies above
the sloping boundary in Eq.~\eqref{eq:analyticRegion}.
The largest unweighted sum is attained at \(A_x\).
Balanced protection is attained where the boundary meets
the ray on which the two fractions
\(r_{\mathrm e}/m_{\mathrm e}\) and
\(r_{\mathrm o}/m_{\mathrm o}\) are equal.

\begin{corollary}[Exact protection costs]
\label{cor:analyticCosts}
For \(\TO\), \(\cTO\), and \(\CGP\), the individual and
summed suprema in absolute coherence units satisfy
\[
M_{\mathrm e}=c(1-x/3),\qquad
M_{\mathrm o}=c(1-x/2),\qquad
S=c(2-x),\qquad
\Delta=\frac{cx}{6}.
\]
The balanced optimum is
\begin{equation}
\alpha_x=\frac{8-4x}{8-3x},
\label{eq:analyticAlpha}
\end{equation}
with selected pair
\((\alpha_xm_{\mathrm e},\alpha_xm_{\mathrm o})\).
Its loss in summed output coherence is
\begin{equation}
c\big[(2-x)-\alpha_x(m_{\mathrm e}+m_{\mathrm o})\big]
=\frac{cx(2-x)}{3(8-3x)}.
\label{eq:analyticBalanceCost}
\end{equation}
The closure extrema are attained. Finite-bath statements are
suprema with exact target populations.
\end{corollary}

For example, at \(x=1/2\), maximizing the sum gives
\((r_{\mathrm e},r_{\mathrm o})=(3/4,3/4)\), whose
smaller normalized score is \(9/10\).
The balanced point is instead \((10/13,9/13)\), where
both coherences retain \(12/13\) of their individual optima.
This improvement costs \(c/26\) in summed output coherence;
see Fig.~\ref{fig:analytic}. At the same final populations,
it raises the guaranteed fraction of the individual phase
Fisher optima from \((9/10)^2\) to \((12/13)^2\).

\begin{figure}[t]
\centering
\begin{tikzpicture}[x=6.2cm,y=5.2cm]
  \draw[->] (0,0) -- (0.90,0) node[below right] {$r_{\mathrm e}$};
  \draw[->] (0,0) -- (0,0.86) node[above left] {$r_{\mathrm o}$};
  \foreach \x/\lab in {0.25/{0.25},0.50/{0.50},0.75/{0.75}}
    \draw (\x,0.008)--(\x,-0.008) node[below] {\scriptsize \lab};
  \foreach \y/\lab in {0.25/{0.25},0.50/{0.50},0.75/{0.75}}
    \draw (0.008,\y)--(-0.008,\y) node[left] {\scriptsize \lab};

  \fill[gray!18]
    (0,0)--(0.83333,0)--(0.83333,0.50)--(0.75,0.75)--(0,0.75)--cycle;
  \draw[thick]
    (0,0)--(0.83333,0)--(0.83333,0.50)--(0.75,0.75)--(0,0.75)--cycle;
  \draw[very thick] (0.75,0.75)--(0.83333,0.50);

  \fill (0.75,0.75) circle (0.012)
    node[above left] {\scriptsize max sum};
  \fill (0.76923,0.69231) circle (0.012)
    node[below left] {\scriptsize balanced};
  \draw[very thick] (0.83333-0.012,0.75-0.012)--(0.83333+0.012,0.75+0.012);
  \draw[very thick] (0.83333-0.012,0.75+0.012)--(0.83333+0.012,0.75-0.012);
  \node[anchor=west] at (0.84,0.75) {\scriptsize separate maxima};
\end{tikzpicture}
\caption{\label{fig:analytic}
Selected retentions at \(x=1/2\). The shaded set is the
region for both CGP and closure-TO. Finite-bath thermal
operations are dense in this set even when the target
populations are required exactly. The two marked points
on the Pareto segment maximize the sum and the balanced
score, respectively. The separate maxima would give
\((5/6,3/4)\), which lies outside the region.}
\end{figure}

\subsection{Further thermal restrictions at positive temperature}
\label{sec:finite}
\label{sec:thermalbounds}

At infinite temperature, the two coherences compete because
their individual optima require different population flows,
while the closed thermal and CGP regions coincide. We now ask
whether a common thermal unitary imposes a further restriction
at positive temperature. Consider the nondegenerate Hamiltonian
\begin{equation}
H=\diag(0,1,3,7),\qquad
\beta=\log(11/10),\qquad q=e^{-\beta}=\frac{10}{11}.
\label{eq:finiteHamiltonian}
\end{equation}
The six positive gaps are distinct. Take
\begin{equation}
p=\left(\frac25,\frac3{10},\frac15,\frac1{10}\right)^T,
\qquad
p'=\left(\frac{19}{50},\frac7{25},\frac{11}{50},\frac3{25}\right)^T,
\qquad
\rho=\diag p+\frac1{20}(\ee03+\ee30+\ee12+\ee21).
\label{eq:finiteTask}
\end{equation}
Both coherent blocks are positive definite and have final
population \(1/2\). We again use \(r_{\mathrm e}\) and
\(r_{\mathrm o}\), now with input amplitude \(c=1/20\).
The complete output diagonal is fixed to \(p'\).

Under the basis identification
\(0=00,1=01,2=10,3=11\), this Hamiltonian is
\[
H=3n_A+n_B+3n_An_B,\qquad
n_A=\ee11\otimes I_2,\quad n_B=I_2\otimes\ee11.
\]
The two coherent blocks are the even- and odd-parity sectors
of an interacting two-qubit register. Their phases can be
read in separate blocks. The restriction concerns how well
both signals survive one thermal change of populations.

As at infinite temperature, two population routes are useful.
The matrix \(P_A\) exchanges populations in pairs \(02,13\),
while \(P_B\) uses \(03,12\). Each pair transfers \(1/50\)
upward for the specified input. Their exact entries are given
in Appendix~\ref{app:positive}; Fig.~\ref{fig:populationRoutes}
shows the pairs involved.
The fixed-\(P\) CGP retention pairs are \((u,v)\) and
\((e,f)\), respectively, where
\[
u=0.777953240551\ldots,\qquad
v=0.818919646779\ldots,\qquad
e=0.855000104272\ldots,\qquad
f=0.617905633714\ldots.
\]
These constants are defined exactly in
Eq.~\eqref{eq:finiteConstants}.

\begin{figure}[t]
\centering
\begin{tikzpicture}[>=stealth]
  \node at (1.55,2.45) {\small route \(P_A\)};
  \foreach \x/\lab in {0/{0},1/{1},2/{2},3/{3}}
    \node[circle,draw,minimum size=6mm] (a\lab) at (\x,1.15) {\lab};
  \draw[<->,thick] (a0) to[bend right=25] (a2);
  \draw[<->,thick] (a1) to[bend right=25] (a3);
  \draw[dashed,thick] (a0) to[bend left=30] node[above] {\scriptsize phase \(03\)} (a3);
  \draw[dashed,thick] (a1) -- (a2);\node at (1.50,1.58) {\scriptsize phase \(12\)};
  \node[align=center] at (1.5,0.20) {\scriptsize population flow: \(0\leftrightarrow2,\;1\leftrightarrow3\)\\[-1mm]
                                    \scriptsize maximum summed retention};

  \begin{scope}[xshift=6.0cm]
    \node at (1.55,2.45) {\small route \(P_B\)};
    \foreach \x/\lab in {0/{0},1/{1},2/{2},3/{3}}
      \node[circle,draw,minimum size=6mm] (b\lab) at (\x,1.15) {\lab};
    \draw[<->,thick] (b0) to[bend right=32] (b3);
    \draw[<->,thick] (b1) to[bend right=25] (b2);
    \draw[dashed,thick] (b0) to[bend left=30] node[above] {\scriptsize phase \(03\)} (b3);
    \draw[dashed,thick] (b1) -- (b2);\node at (1.50,1.58) {\scriptsize phase \(12\)};
    \node[align=center] at (1.5,0.20) {\scriptsize population flow: \(0\leftrightarrow3,\;1\leftrightarrow2\)\\[-1mm]
                                      \scriptsize maximum \(03\) retention};
  \end{scope}
\end{tikzpicture}
\caption{\label{fig:populationRoutes}
Population routes \(P_A\) and \(P_B\) for the same final
diagonal in the positive-temperature example. Solid links
show which pairs exchange populations in the two
constructions; dashed links mark the coherences \(03\)
and \(12\). The sum supremum is approached using \(P_A\),
whereas the individual \(03\) supremum is approached
using \(P_B\). Balanced protection requires a different
choice.}
\end{figure}

\begin{theorem}[Exact positive-temperature benchmarks]
\label{thm:finiteExact}
For the task in Eq.~\eqref{eq:finiteTask}, finite-TO,
closure-TO, and CGP have the same individual and
unweighted-sum suprema:
\begin{equation}
\sup r_{\mathrm e}=e,\qquad
\sup r_{\mathrm o}=v,\qquad
\sup(r_{\mathrm e}+r_{\mathrm o})=u+v.
\label{eq:finiteExact}
\end{equation}
The closure extrema are attained. Finite-bath channels
approach them with the target populations exact at every step.
Consequently,
\begin{equation}
\Delta_{\TO}=\Delta_{\cTO}
=\frac{e-u}{20}
=0.003852343186\ldots .
\label{eq:finiteGap}
\end{equation}
Every CGP sum maximizer has population matrix \(P_A\)
and selected pair \((u,v)\). The same holds in closure-TO;
every exact-target finite-TO sum-maximizing sequence
approaches that pair.
\end{theorem}

For the route \(P_B\), two geometric Gibbs ladders approach
the individual \(03\) optimum. The sum-optimal route
\(P_A\) uses a shared Gibbs bath: its rotations coordinate
the stay amplitudes for the two coherences. Truncating the
bath and correcting two rotation strengths preserves the
population matrix exactly while the coherences approach
\((u,v)\). Appendix~\ref{app:positive} gives these
constructions and the population bounds proving optimality.

Although the individual and maximum-sum values coincide for
TO and CGP, other joint objectives distinguish the two classes. The
complete CGP Pareto boundary consists of two arcs joining
\((u,v)\) to \((e,f)\). Every boundary point fixes its
population matrix uniquely and saturates the two scalar
coherence bounds. At intermediate points, that matrix has
a pattern of vanishing transitions incompatible with such
simultaneous saturation by a common thermal unitary.
The resulting exclusion holds even in the closure:
\begin{equation}
\mathcal R_{\cTO}\cap
\operatorname{Pareto}(\mathcal R_{\CGP})
=\{(u,v),(e,f)\}.
\label{eq:thermalCGPBoundary}
\end{equation}
The population arcs, the zero-pattern condition, and its
application are given in Appendix~\ref{app:positive}.
The condition concerns simultaneous saturation under its
stated hypotheses; it does not impose detailed balance on
general thermal operations.

The unique balanced CGP point gives
\[
\alpha_{\CGP}^\star=0.93947683077729995098\ldots.
\]
It lies between the two endpoints and is therefore excluded
from closure-TO. This example separates the extra restriction
of a thermal dilation from the competition between population
matrices already present at infinite temperature.

To bound what a thermal operation can achieve between the
known endpoints, we impose moment constraints on the blocks
of one energy-preserving unitary. The resulting semidefinite
relaxation gives rigorously certified upper bounds valid for
arbitrary bath dimensions. An explicit finite-bath
construction supplies a lower bound within a small interval
of the upper bound. The normalizers in the balanced score are
the exact individual thermal suprema \(e\) and \(v\).

\begin{theorem}[Certified balanced thermal protection]
\label{thm:balancedThermal}
For Eq.~\eqref{eq:finiteTask},
\begin{equation}
0.93693165<
\alpha^\star_{\TO}
=\alpha^\star_{\cTO}
<0.938615.
\label{eq:balancedInterval}
\end{equation}
There is a finite thermal operation \(\Phi_{\rm bal}\)
with bath dimension at most 266, exact target populations,
and
\begin{equation}
r_{\mathrm e}(\Phi_{\rm bal})>0.80107666625567,\qquad
r_{\mathrm o}(\Phi_{\rm bal})>0.76727174335430.
\label{eq:balancedWitness}
\end{equation}
Its balanced score satisfies
\begin{equation}
\frac{\alpha(\Phi_{\rm bal})}{\alpha^\star_{\TO}}>0.9982.
\label{eq:nearOptimal}
\end{equation}
\end{theorem}

The upper bound follows from a certified inequality for
\((11/5)r_{\mathrm e}+r_{\mathrm o}\). The finite
construction uses a 202-dimensional Gibbs bath together
with population corrections, giving a total bath dimension
at most 266. The correction makes the target populations
exact. The moment relaxation, the construction, and the
proof of Eq.~\eqref{eq:balancedInterval} are given in
Appendix~\ref{app:thermal}.

Maximizing the unweighted sum gives the pair \((u,v)\)
and smallest normalized retention
\[
\frac ue=0.909886720088\ldots.
\]
The balanced construction instead guarantees more than
\(0.93693165\) of each individual optimum at the same
target populations. This guarantee entails a loss of more
than \(0.00119853\) in summed absolute coherence relative
to the sum supremum; the bound is derived in
Appendix~\ref{app:thermal}.
For the phase family, \(V_\ell=4cr_\ell\) and
\(\mathcal F_{\ell\ell}=8c^2r_\ell^2\).
The construction therefore attains more than \(99.64\%\)
of the optimal smallest fraction of the individual phase
Fisher optima. Figure~\ref{fig:positive} shows the finite
construction, the thermal upper bound, and the excluded
CGP balanced point.

\begin{figure}[t]
\centering
\begin{tikzpicture}[x=1cm,y=1cm]
  \draw[->] (0,0) -- (6.65,0) node[below right] {$r_{\mathrm e}$};
  \draw[->] (0,0) -- (0,5.18) node[above left] {$r_{\mathrm o}$};
  \foreach \x/\lab in {0.64/{0.78},1.92/{0.80},3.20/{0.82},4.48/{0.84},5.76/{0.86}}
    \draw (\x,0.06)--(\x,-0.06) node[below] {\scriptsize \lab};
  \foreach \y/\lab in {0.435/{0.62},1.304/{0.66},2.174/{0.70},3.043/{0.74},3.913/{0.78},4.783/{0.82}}
    \draw (0.06,\y)--(-0.06,\y) node[left] {\scriptsize \lab};

  \draw[thick] (0.509,4.759) -- (0.890,4.510) -- (1.268,4.260) -- (1.644,4.009) -- (2.017,3.757) -- (2.387,3.505) -- (2.754,3.252) -- (3.119,2.999) -- (3.481,2.744) -- (3.841,2.489) -- (4.043,2.227) -- (4.245,1.965) -- (4.446,1.703) -- (4.646,1.440) -- (4.846,1.178) -- (5.045,0.915) -- (5.243,0.652) -- (5.440,0.389);
  \draw[thick,dashed] (0.509,4.759) -- (1.694,3.884) -- (1.989,3.636) -- (2.056,3.579) -- (5.440,0.389);

  \draw[gray!70,very thick]
    (0.45,4.8851363043) -- (3.25,2.7927450000);

  \draw[dotted,thick]
    (0.95,3.2983444735) -- (3.05,3.9815582450);

  \fill (1.9889066404,3.6363422468) circle (0.026);
  \draw[thin] (1.972,3.612) -- (1.40,3.18);
  \node[anchor=east] at (1.35,3.18) {\scriptsize finite TO};

  \draw[thick] (2.0630134988,3.6483082821) --
               (2.0990134988,3.6843082821);
  \draw[thick] (2.0630134988,3.6843082821) --
               (2.0990134988,3.6483082821);
  \draw[thin] (2.105,3.647) -- (3.10,3.12);
  \node[anchor=west] at (3.16,3.12)
     {\scriptsize thermal upper bound};

  \draw[fill=white,thin] (2.1281784497,3.6816529221) circle (0.020);
  \draw[thin] (2.144,3.697) -- (2.68,4.03);
  \node[anchor=west] at (2.73,4.03) {\scriptsize CGP balanced};
  \node[anchor=west] at (3.75,4.55) {\scriptsize solid: CGP boundary};
  \node[anchor=west] at (3.75,4.25) {\scriptsize dashed: explicit TO inner bound};
\end{tikzpicture}
\caption{\label{fig:positive}
Positive-temperature retentions near the balanced direction.
The solid curve shows the CGP Pareto boundary, and the
dashed line is an inner bound from explicit thermal
constructions. The gray line plots the rounded certified
bound \((11/5)r_{\mathrm e}+r_{\mathrm o}=2.53418502\);
it is not claimed to touch the thermal region.
The dotted line is the balanced ray
\(r_{\mathrm o}=(v/e)r_{\mathrm e}\).
The filled point shows the finite-bath balanced guarantee,
the cross marks the upper bound along the balanced ray,
and the open point is the balanced CGP optimum, which is
excluded from closure-TO. The thermal frontier between
the known endpoints has not been determined exactly.}
\end{figure}

The separation also survives a relaxation of the population
constraint. Let \(\sigma_\star\) be the balanced CGP output
with positive real coherences. A witness involving the two
coherences and the output populations gives
\begin{equation}
\frac12\|\Phi(\rho)-\sigma_\star\|_1>0.00005025
\qquad\text{for every }\Phi\in\cTO .
\label{eq:stateSeparation}
\end{equation}
Here no final-population constraint is imposed on the thermal
competitor. Appendix~\ref{app:robust} gives the population-tolerance
bound and describes the witness. Its full specification and
certified comparisons are given in the Supplemental Material~\cite{SM}.

\section{Discussion}
\label{sec:discussion}

The preservation of several coherences by one thermal process
is a compatibility problem. In an energy-resolved representation,
each elementary coherence transfer is an overlap of bath
amplitudes. The conditions for attaining several bounds must
hold for the same vectors and the same energy-preserving
unitary, together with trace and Gibbs preservation. This
description applies to arbitrary finite-dimensional Hamiltonians.
It supplies necessary tests for proposed saturation patterns;
the actual individual optima must still be compared with what
one process can achieve jointly.

The mixed-degenerate example establishes such an incompatibility
using separately feasible thermal operations and a certified
upper bound on their joint performance. Hence the freedom
provided by degenerate energy levels does not guarantee that
the individual optima are compatible. Determining when
degeneracy improves coherence preservation is a further
question: the example does not make a controlled comparison
with a corresponding Hamiltonian whose degeneracy has been
lifted.

For the two-coherence families, the quantitative results separate
competition between population matrices from the additional
restriction of a thermal dilation. At infinite temperature,
the complete closed region is determined by the first effect.
At positive temperature, the intermediate CGP boundary is
excluded from closure-TO, although the individual and
maximum-sum values agree. Maximizing the sum and maximizing
the smallest fraction of the individual optima then select
different processes. When the coherences carry independent
phases and the final populations are fixed, the choice has a
direct meaning for the visibility and phase Fisher information
of the two signals.

The finite-bath constructions keep the target populations
exact. They approach the infinite-temperature boundary and
the positive-temperature endpoint suprema, while a finite
bath gives balanced protection within \(0.18\%\) of the
thermal optimum. The exact intermediate positive-temperature
frontier and balanced optimum remain to be determined.
For the latter, a matching upper bound for one weighted
thermal objective and a construction on the balanced ray
would suffice; see Eq.~\eqref{eq:balancedSupport}.
Stronger moment constraints and other shared-bath constructions
may reduce the remaining interval. The bath dimensions used
here are sufficient, and have not been proved minimal.
Extending the selected-target results beyond forests also
requires treating phase constraints around cycles.

\section*{Data Availability Statement}
The appendices give the channel reductions, derivations, and
bath constructions used in the main text. The accompanying
Supplemental Material contains further proofs and complete
specifications of the finite constructions. The directory
\path{https://github.com/studzinskimg-afk/coherences-transfer/tree/main/numerics_arxiv} contains the fixed Choi, population,
and thermal-moment certificate data, rational shell-rotation
records, independent verification programs, and execution logs.
The programs verify the deposited data without rerunning
numerical optimization. A statement-to-proof map identifies
the supporting result for each main theorem.

\section*{Acknowledgments}
This work is carried out under IRA Programme, project no. FENG.02.01-IP.05-0006/23, financed by the FENG
program 2021--2027, Priority FENG.02, Measure FENG.02.01, with the support of the FNP. ChatGPT (OpenAI) assisted with proof-checking derivations, preparing
numerical and consistency scripts, and editing the manuscript. Its use was guided by authors-specified prompts concerning the mathematical assumptions, arguments, and statements under consideration. The authors are responsible for independently examining, verifying, and revising all AI-assisted material and assume full responsibility for the accuracy and scientific content of the final work.

\clearpage
\appendix

\section{General framework and transition amplitudes}
\label{app:general}

\subsection{Thermal operations and population constraints}

Let \(H_S=\sum_\alpha E_\alpha\Pi_\alpha\) act on a
finite-dimensional Hilbert space. The projectors \(\Pi_\alpha\)
refer to distinct energies and may have rank greater than one.
We use \(i,j,r,s\) for energy eigenbasis labels and write \(E_i\)
for the energy of \(\ket i\). Thus, within a degenerate
eigenspace, \(i=(\alpha,a)\) includes a degeneracy index.

A finite-bath thermal operation at inverse temperature \(\beta\)
is given by Eq.~\eqref{eq:TOdef}.
We denote this finite-bath class by \(\TO=\TOfin\) and its
channel-norm closure by \(\cTO\). Since the system dimension
is fixed, the choice of norm does not affect the closure.
We shall distinguish a supremum over finite baths from its
attainment: a sequence in \(\TO\) may approach a value that
no finite bath attains. When optimizing over \(\cTO\), we
impose the population constraint on the limiting channel.
For a characterization of thermal operations through their
dilations, see Ref.~\cite{LieEtAl2026}.

Every thermal operation is covariant,
\begin{equation}
\Phi(e^{-itH_S}Xe^{itH_S})
=e^{-itH_S}\Phi(X)e^{itH_S},
\label{eq:covariance}
\end{equation}
and preserves
\(\gamma=e^{-\beta H_S}/\Tr e^{-\beta H_S}\).
The closed class of covariant Gibbs-preserving CPTP maps is
denoted by \(\CGP\), so that
\begin{equation}
\TO\subseteq\cTO\subseteq\CGP .
\label{eq:TOsubsetCGP}
\end{equation}
We do not impose Markovian dynamics.

The basis-state transition probabilities are
\begin{equation}
P_{ri}:=\bra r\Phi(\ee ii)\ket r .
\label{eq:populationMatrix}
\end{equation}
With \(g_i=e^{-\beta E_i}\), they satisfy
\begin{equation}
P\ge0,\qquad \one^TP=\one^T,\qquad Pg=g .
\label{eq:populationBasic}
\end{equation}
For a fixed input basis state, summing over \(r\) in an
eigenspace gives the probability of reaching that eigenspace.
For the input \(\Pi_\lambda/\rank\Pi_\lambda\), the corresponding
coarse transition probability is
\[
\frac{1}{\rank\Pi_\lambda}
\sum_{\substack{i\in\lambda\\r\in\alpha}}P_{ri}.
\]
The input convention is necessary when an eigenspace is
degenerate.

For a full final diagonal \(p'\), write
\begin{equation}
\Gamma_{\mathfrak F}(\rho,p')
=\{\Phi\in\mathfrak F:\diag\Phi(\rho)=p'\}.
\label{eq:fullFeasible}
\end{equation}
If the Hamiltonian is nondegenerate, the population constraint
is \(Pp=p'\), where \(p=\diag\rho\).
A different task fixes only the occupations of distinct energy
eigenspaces:
\begin{equation}
\Gamma(\rho,p^f;\mathfrak F)
=\{\Phi\in\mathfrak F:
\Tr[\Pi_\alpha\Phi(\rho)]=p_\alpha^f\ \text{for all }\alpha\}.
\label{eq:coarseFeasible}
\end{equation}
In the mixed-degenerate example we fix only the total
occupation of each energy eigenspace, as in
Eq.~\eqref{eq:coarseFeasible}. The two nondegenerate examples
use the full output diagonal in Eq.~\eqref{eq:fullFeasible}.
Prescribing \(P=P_0\) is a further restriction: it fixes the
transition probabilities for every input basis state.

Time-translation covariance preserves the Bohr modes
\begin{equation}
\mathcal M_\omega
=\Span\{\ee ij:E_i-E_j=\omega\}.
\label{eq:bohrModes}
\end{equation}
If \(E_i\ne E_j\) and the mode \(\mathcal M_{E_i-E_j}\)
is one dimensional, covariance gives
\begin{equation}
\Phi(\ee ij)=\kappa_{ij}\ee ij.
\label{eq:scalarAction}
\end{equation}
For a nondegenerate Hamiltonian with one-dimensional nonzero
Bohr modes, this holds for every \(i\ne j\).
Note that degeneracy of an energy eigenspace and degeneracy
of a Bohr mode are different notions. The block
\(\Pi_\alpha X\Pi_\lambda\) has operator-space dimension
\(\rank\Pi_\alpha\rank\Pi_\lambda\), and different blocks
belong to the same mode whenever their energy gaps coincide.

\subsection{Individual targets, summed retention, and balanced protection}

Let \(\mathcal F\) be a nonempty feasible set, with the channel
class and population constraints fixed. For each selected
ordered pair \((i_k,j_k)\), define
\begin{equation}
z_k(\Phi)=\Phi(\rho)_{i_kj_k},\qquad
Q_k(\Phi)=|z_k(\Phi)|,\qquad
L_k^{(\varphi_k)}(\Phi)
=\operatorname{Re}[e^{-i\varphi_k}z_k(\Phi)].
\label{eq:scores}
\end{equation}
We shall distinguish prescribing a complex value
\(z_k=z_k^{\rm tar}\), requiring a magnitude \(Q_k\ge a_k\),
and requiring a component in a fixed direction
\(L_k^{(\varphi_k)}\ge a_k\). The phase is specified in
the complex-target and directional tasks, but not in the
magnitude task.

For a chosen family of scores \(q_k\), put
\begin{align}
M_k&=\sup_{\Phi\in\mathcal F}q_k(\Phi),&
S_w&=\sup_{\Phi\in\mathcal F}\sum_k w_kq_k(\Phi),\\
\Delta_w&=\sum_k w_kM_k-S_w,\qquad &&w_k>0 .
\label{eq:summedGap}
\end{align}
The individual suprema and the sum supremum are all taken
over this same \(\mathcal F\). For nonnegative scores with
\(M_k>0\), we also define
\begin{equation}
\alpha_{\mathcal F}
=\sup_{\Phi\in\mathcal F}\min_k\frac{q_k(\Phi)}{M_k}.
\label{eq:jointScore}
\end{equation}
For an individual channel in this feasible set, we write
\(\alpha(\Phi)=\min_k q_k(\Phi)/M_k\).
This is the balanced protection score: we compare each
achieved score with its own optimum and maximize the smallest
ratio. Here \(M_k\) is the actual supremum over \(\mathcal F\),
not a possibly larger elementary upper bound. In contrast,
\(S_w\) maximizes the weighted sum of the achieved scores.

\begin{lemma}[A certificate for joint incompatibility]
\label{lem:gapWitness}
For finite \(M_k\), \(\Delta_w\ge0\), and \(\Delta_w=0\) if
and only if a sequence in \(\mathcal F\) approaches all \(M_k\)
simultaneously. Suppose \(W_k\in\mathcal F\), and suppose
\(\sum_k w_kq_k(\Phi)\le U\) on a set containing
\(\mathcal F\). Then
\begin{equation}
\Delta_w\ge\sum_k w_kq_k(W_k)-U .
\label{eq:witnessGap}
\end{equation}
\end{lemma}
\begin{proof}
Every deficit \(d_k(\Phi)=M_k-q_k(\Phi)\) is nonnegative, and
\[
\Delta_w=\inf_{\Phi\in\mathcal F}\sum_k w_kd_k(\Phi).
\]
If the infimum is zero, choose \(\Phi_n\) with summed deficit
less than \(1/n\). Then \(d_k(\Phi_n)<1/(nw_k)\) for every \(k\).
Conversely, convergence of all finitely many deficits to zero
gives a vanishing summed deficit. Finally,
\(M_k\ge q_k(W_k)\) and \(S_w\le U\), which give
Eq.~\eqref{eq:witnessGap}.
\end{proof}

Thus \(\Delta_w>0\) rules out a sequence approaching all
individual suprema, even though it is certified by a
weighted-sum bound.

\subsection{Phase alignment for the selected graphs}

The selected pairs form an undirected graph on the basis labels.
In the examples considered here this graph is a forest: a path in the
mixed-degenerate example and two disjoint edges in the
nondegenerate examples.

\begin{proposition}[Forest objectives]
\label{prop:forest}
For a forest of selected pairs, postcomposition by finite-bath
thermal operations can assign arbitrary phases to the selected
entries and decrease their magnitudes independently, without
changing any population. Consequently the individual and
nonnegative weighted-sum suprema for magnitudes equal those for
phase-aligned real parts. The same holds for balanced protection.

At positive temperature the selected nonnegative retention region
is convex for \(\TO\). It is convex for \(\cTO\) at every
temperature, and for \(\CGP\).
\end{proposition}
\begin{proof}
Choose a root in each tree. Starting at the root, assign vertex
phases recursively so that the phase difference across each edge
has the desired value. A tree has one path from the root to each
vertex, so no conflicting phase equation occurs. The resulting
diagonal unitary commutes with \(H_S\).

To attenuate one edge by \(\lambda\in[0,1]\), delete that edge.
Its tree splits into two parts. Average the two diagonal unitaries
which give one part phases \(e^{\pm i\theta}\), where
\(\cos\theta=\lambda\), and give the other part phase one.
A degenerate Gibbs qubit realizes this equal mixture.
Only the selected deleted edge crosses the cut. Repeating the
construction proves independent attenuation.

Phase-align two channel outputs before mixing them.
Their selected nonnegative entries then mix linearly.
At \(\beta>0\), arbitrary mixing probabilities can be
realized by a Gibbs flag; at \(\beta=0\), a uniform finite flag
realizes rational probabilities and their limits give all
mixtures in the closure. Detailed bath constructions are
given in Supplemental Material~\cite{SM}.
\end{proof}

At fixed output populations this proposition makes exact complex
targets on the selected forests equivalent to magnitude
attainability followed by phase assignment. The statement is
not used for graphs containing cycles.

\subsection{Kraus--Gram and thermal-shell representations}
\label{sec:TO-shell}

We begin with the common amplitude representation. Let
\(\Phi(X)=\sum_a K_aXK_a^\dagger\), and define
\[
v_{r|i}=\sum_a(K_a)_{ri}\ket a .
\]
Direct expansion gives
\begin{equation}
P_{ri}=\|v_{r|i}\|^2,\qquad
\Phi(\ee ij)_{rs}=\langle v_{s|j},v_{r|i}\rangle .
\label{eq:genericKrausGram}
\end{equation}
Thus complete positivity implies
\begin{equation}
|\Phi(\ee ij)_{rs}|^2\le P_{ri}P_{sj}.
\label{eq:elementaryBound}
\end{equation}
For two nonzero vectors, equality holds precisely when they are
complex-collinear.

For a thermal operation, we use an energy-resolved Stinespring
representation \cite{Stinespring1955PositiveFunctions,
AwZawBalanzoJuandoScarani2024Dilations}. Let
\(H_B=\sum_\mu\epsilon_\mu\ee{\mu}{\mu}\) and
\(\tau_B=\sum_\mu q_\mu\ee{\mu}{\mu}\).
The natural Kraus operators are
\[
K_{\nu\mu}=\sqrt{q_\mu}\,{}_B\langle\nu|U|\mu\rangle_B,
\]
and their entries vanish unless
\(E_i+\epsilon_\mu=E_r+\epsilon_\nu\).
The corresponding vectors are
\begin{equation}
w_{r|i}
=\sum_{\mu,\nu}
\sqrt{q_\mu}\,U_{r\nu,i\mu}\ket{\mu,\nu},
\qquad
U_{r\nu,i\mu}=0
\ \text{if }E_i+\epsilon_\mu\ne E_r+\epsilon_\nu .
\label{eq:shellAmplitudeVector}
\end{equation}
All vectors use the same bath input-output basis.

\begin{theorem}[Thermal transition amplitudes]
\label{thm:TOinnerproducts}
For a thermal operation,
\[
P_{ri}=\|w_{r|i}\|^2,\qquad
\Phi(\ee ij)_{rs}=\langle w_{s|j},w_{r|i}\rangle .
\]
The second quantity vanishes unless \(E_i-E_j=E_r-E_s\).
If \(P_{ri}P_{sj}>0\), equality in
Eq.~\eqref{eq:elementaryBound} is equivalent to
\(w_{r|i}\parallel w_{s|j}\).
\end{theorem}
\begin{proof}
The population and coherence expressions are, respectively,
\[
\sum_{\mu,\nu}q_\mu|U_{r\nu,i\mu}|^2,\qquad
\sum_{\mu,\nu}q_\mu U_{r\nu,i\mu}
                  \overline{U_{s\nu,j\mu}}.
\]
The energy equations for a nonzero term in the second sum
subtract to \(E_i-E_j=E_r-E_s\). The sums are the norm and
inner product in Eq.~\eqref{eq:shellAmplitudeVector}.
Cauchy--Schwarz and its equality criterion complete the proof.
\end{proof}

For a state \(\rho\), this gives
\begin{align}
\Phi(\rho)_{rs}
&=\sum_{E_i-E_j=E_r-E_s}\rho_{ij}
       \langle w_{s|j},w_{r|i}\rangle,\\
|\Phi(\rho)_{rs}|
&\le\sum_{E_i-E_j=E_r-E_s}|\rho_{ij}|
                  \sqrt{P_{ri}P_{sj}} .
\label{eq:modeTransportBound}
\end{align}
Equality requires collinearity in every nonzero contributing
term and a common phase for all summands, with all amplitude
vectors obtained from the same thermal unitary.

Trace and Gibbs preservation also give
\begin{equation}
\sum_r\langle w_{r|j},w_{r|i}\rangle=\delta_{ij},
\qquad
\sum_i\gamma_i\langle w_{s|i},w_{r|i}\rangle
=\gamma_r\delta_{rs}.
\label{eq:gramConservation}
\end{equation}
The same identities hold for the generic Kraus vectors of a
trace- and Gibbs-preserving map. To obtain
Eq.~\eqref{eq:previewPolygon}, take distinct input labels
\(i,j\) and set \(z_r=\Phi(\ee ij)_{rr}\),
\(a_r=\sqrt{P_{ri}P_{rj}}\). The trace identity gives
\(\sum_r z_r=0\), so for every \(r_0\),
\[
|z_{r_0}|=\left|\sum_{r\ne r_0}z_r\right|
\le\sum_{r\ne r_0}|z_r|
\le\sum_{r\ne r_0}a_r.
\]
If all elementary bounds saturate, then
\(a_{r_0}\le\sum_{r\ne r_0}a_r\) for every
\(r_0\), which is Eq.~\eqref{eq:previewPolygon}.
For distinct output labels \(r,s\), the Gibbs identity
similarly gives a zero sum of the terms
\(\gamma_i\Phi(\ee ii)_{rs}\), whose magnitudes are
bounded by \(\gamma_i\sqrt{P_{ri}P_{si}}\).
Covariance forces the trace terms to vanish when
\(E_i\ne E_j\), and the Gibbs terms to vanish when
\(E_r\ne E_s\). This explains the zero-frequency
qualification in the main text. Full polygon closure
conditions, the stronger individual bounds, and the auxiliary
rank-deficient Gram constraints are proved in the
Supplemental Material~\cite{SM}.
The auxiliary analysis uses range inclusion for positive blocks
\cite{Douglas1966} and, in its conditional non-genericity statement,
the zero-set property of nontrivial real-analytic functions
\cite{Mityagin2020}.

\subsection{A constructive CGP solution at fixed population dynamics}

Suppose now that the Hamiltonian is nondegenerate and all its
nonzero Bohr modes are one dimensional. At fixed \(P\), all
scalar retention bounds can be saturated by
one CGP channel \cite{LostaglioKorzekwaMilne2017}.
The construction is useful because it separates population
optimization from fixed-\(P\) compatibility.

\begin{proposition}[Fixed-population retention]
\label{prop:fixedP}
Let \(P\) satisfy Eq.~\eqref{eq:populationBasic}. The operators
\begin{equation}
K_0=\sum_i\sqrt{P_{ii}}\,\ee ii,\qquad
K_{ri}=\sqrt{P_{ri}}\,\ee ri\quad(r\ne i)
\label{eq:fixedPKraus}
\end{equation}
define a CGP channel with population matrix \(P\) and
\begin{equation}
\Phi_P(\ee ij)=\sqrt{P_{ii}P_{jj}}\,\ee ij,\qquad i\ne j.
\label{eq:fixedPDamping}
\end{equation}
For simple nonzero modes, it attains all scalar
magnitude-retention bounds at that \(P\).
\end{proposition}
\begin{proof}
The sum of \(K_a^\dagger K_a\) is
\(\sum_i(\sum_rP_{ri})\ee ii=I\).
The jump operators have definite energy shifts, while \(K_0\)
commutes with \(H_S\); hence the channel is covariant.
Its diagonal action is \(P\), so \(Pg=g\) proves Gibbs
preservation. For \(i\ne j\), every jump Kraus contribution
vanishes, and \(K_0\) gives Eq.~\eqref{eq:fixedPDamping}.
Equation~\eqref{eq:elementaryBound} proves optimality.
\end{proof}

Let
\[
\mathcal P(p,p')=\{P\ge0:\one^TP=\one^T,\ Pg=g,\ Pp=p'\}.
\]
For selected simple modes and \(w_k\ge0\),
\begin{equation}
\max_{\Phi\in\Gamma_{\CGP}(\rho,p')}
\sum_k w_k|\Phi(\rho)_{i_kj_k}|
=
\max_{P\in\mathcal P(p,p')}
\sum_k w_k|\rho_{i_kj_k}|\sqrt{P_{i_ki_k}P_{j_kj_k}}.
\label{eq:populationReduction}
\end{equation}
The upper inequality follows from the scalar bounds; the reverse
inequality follows from Proposition~\ref{prop:fixedP}.
Introducing nonnegative variables \(t_k\) with
\[
t_k^2\le|\rho_{i_kj_k}|^2P_{i_ki_k}P_{j_kj_k}
\]
gives a second-order cone formulation.

For a forest of selected simple modes, prescribed complex targets
\(z_{ij}\) are attainable in CGP if and only if one
\(P\in\mathcal P(p,p')\) satisfies
\begin{equation}
|z_{ij}|^2\le|\rho_{ij}|^2P_{ii}P_{jj}
\quad\text{on every selected edge}.
\label{eq:exactForestTargets}
\end{equation}
Necessity is the scalar bound. Sufficiency follows from the
fixed-\(P\) channel and the independent attenuation and phase
assignment in Proposition~\ref{prop:forest}; zero input entries
require zero targets. The equivalent positive-Gram construction
is given in Supplemental Material~\cite{SM}.
This is related to positive completion of partially specified
Hermitian matrices \cite{GroneEtAl1984}.

\subsection{Population choice and fixed-matrix compatibility}

For a feasible set \(\mathcal F\), let
\(\mathcal F_P=\{\Phi\in\mathcal F:P(\Phi)=P\}\), restricted
to population matrices realized in \(\mathcal F\). Put
\[
m_k(P)=\sup_{\Phi\in\mathcal F_P}q_k(\Phi),\quad
s_w(P)=\sup_{\Phi\in\mathcal F_P}\sum_k w_kq_k(\Phi),\quad
d_w(P)=\sum_k w_km_k(P)-s_w(P).
\]
Then
\begin{equation}
\Delta_w
=\inf_P\left\{
\sum_k w_k[M_k-m_k(P)]+d_w(P)
\right\}.
\label{eq:fiberDecomposition}
\end{equation}
Indeed, the expression in braces is
\(\sum_k w_kM_k-s_w(P)\), and
\(\sup_Ps_w(P)=S_w\).
Both terms are nonnegative.

For the CGP task in Eq.~\eqref{eq:populationReduction},
\(d_w(P)=0\): at a prescribed \(P\), all selected retentions
can be saturated together. Nevertheless, different objectives
can favor different matrices \(P\). For TO, \(d_w(P)\) can also be positive. Both losses must
be evaluated at the same \(P\) before taking the infimum.

\section{The mixed-degenerate certificate and thermal witnesses}
\label{app:mixed}

The system, input, and target occupations are given in Eqs.~\eqref{eq:certificateGibbs} and~\eqref{eq:certificateState}.

\subsection{Positivity of the input state}

The input has unit trace and diagonal-dominance margins
\(0.26,0.20,0.15,0.11\). Since
\(2|x_i||x_j|\le|x_i|^2+|x_j|^2\),
\[
x^\dagger\rho x
\ge
\sum_i\left(\rho_{ii}-\sum_{j\ne i}|\rho_{ij}|\right)|x_i|^2
\ge\frac{11}{100}\|x\|^2.
\]
Thus \(\rho\ge(11/100)I>0\).

\subsection{Choi semidefinite program}

With output index \(x\) and input index \(u\), we use the
Choi representation \cite{Choi1975CPMaps}
\begin{equation}
J(\Phi)=\sum_{x,y,u,v}\Phi_{xy,uv}\ket{x,u}\bra{y,v},
\qquad
\Phi(X)_{xy}=\sum_{u,v}J_{xu,yv}X_{uv}.
\label{eq:certificateChoiConvention}
\end{equation}
Complete positivity, trace preservation, Gibbs preservation, and
covariance are
\begin{equation}
\begin{aligned}
J&\ge0,\\
\sum_xJ_{xu,xv}&=\delta_{uv},\\
\sum_uJ_{xu,yu}\gamma_u&=\gamma_x\delta_{xy},\\
J_{xu,yv}&=0\quad\text{unless }E_x-E_u=E_y-E_v .
\end{aligned}
\label{eq:certificateSDPConstraints}
\end{equation}
The three target populations are linear in \(J\), for example
\[
\sum_{u,v}J_{0u,0v}\rho_{uv}=0.50,\qquad
\sum_{u,v}(J_{1a\,u,1a\,v}+J_{1b\,u,1b\,v})\rho_{uv}=0.35.
\]
The remaining target occupation is \(0.15\). Since the
feasible matrices are positive, have \(\Tr J=4\), and form
a closed set, the optimization is over a compact set.
The primal witnesses below show that this set is nonempty.

For comparison, Table~\ref{tab:mixedScores} gives the values
obtained in separate CGP optimizations alongside those
achieved by the thermal processes constructed below.

\begin{table}[t]
\centering
\begin{tabular}{lcc}
\toprule
Objective & Numerical CGP optimum & Selected TO witness score\\
\midrule
\(L_1\) & \(0.03223447758\) & \(0.027243811688\ldots\)\\
\(L_2\) & \(0.04260690178\) & \(0.042606898032\ldots\)\\
\(L_3\) & \(0.02816833801\) & \(0.026492574688\ldots\)\\
\bottomrule
\end{tabular}
\caption{\label{tab:mixedScores}
Separate optimizations for the task in Eq.~\eqref{eq:certificateState}.
Each TO score comes from a different explicitly specified
process. The incompatibility proof uses certified bounds
rather than the rounded values shown here.}
\end{table}

\begin{lemma}[Certified Choi bounds]
\label{lem:certifiedSDP}
For the mixed-degenerate CGP feasible set,
\begin{equation}
\sum_{k=1}^3M_k^{\CGP}>0.1030097,\qquad
S^{\CGP}<0.088126375461<0.0881264.
\label{eq:mixedCertifiedBounds}
\end{equation}
The inequalities hold for both the directional objectives
in Eq.~\eqref{eq:mixedObjectives} and their magnitudes.
\end{lemma}
\begin{proof}
All data are real. If \(J=A+iB\ge0\), with \(A\) real
symmetric and \(B\) real antisymmetric, then \(x^TAx\ge0\)
for every real \(x\). For \(z=x+iy\),
\(z^\dagger Az=x^TAx+y^TAy\ge0\), hence \(A\ge0\).
Taking real parts preserves the constraints and objectives.

On the covariance support, write the eleven independent equalities
as \(\langle A_i,J\rangle=b_i\). Define
\(\mathcal A(J)=(\langle A_i,J\rangle)_{i=1}^{11}\) and
\(G_{ij}=\Tr(A_iA_j)\).
The deposited dyadic matrices \(\widehat J_k\) are corrected by
\[
J_{1,k}=\widehat J_k+\sum_i
[G^{-1}(b-\mathcal A(\widehat J_k))]_iA_i.
\]
The equalities then hold algebraically. Mixing with an independently
corrected interior point by \(s=2^{-24}\) gives \(J_k\).
Verified blockwise interval \(LDL^{\mathsf T}\) factorizations
prove positivity of \(G\) and of all \(J_k\).
The exact lower endpoints of their three selected scores sum
to more than \(0.1030097\).

For the summed objective \(C_{\rm sum}\), a fixed dyadic vector
\(y\) and \(\delta=2^{-30}\) satisfy
\[
\sum_i y_iA_i-C_{\rm sum}\succeq-\delta I.
\]
For every feasible \(J\),
\[
\langle C_{\rm sum},J\rangle
\le y\cdot b+\delta\Tr J=y\cdot b+4\delta
<0.088126375461.
\]
The fixed data, equality-row order, and independent
integer-interval verification are given in the Supplemental Material~\cite{SM}.
The Choi witnesses were generated with CVXPY and Clarabel
\cite{DiamondBoyd2016CVXPY,GoulartChen2026Clarabel}
and checked by rigorous interval methods
\cite{Johansson2017Arb}; the verification does not rerun
an optimizer. Proposition~\ref{prop:forest} transfers
the bounds to the magnitude objectives.
\end{proof}

\subsection{Finite thermal witnesses for the same population task}

The required witnesses use five-level Gibbs baths.
For a pair \((0,j)\) of gap \(\Delta\), let
\[
q=e^{-\beta\Delta},\qquad
w_n=\frac{q^n}{\sum_{k=0}^4q^k},\qquad
A(q)=\frac{1-q^4}{1-q^5}.
\]
Rotate each equal-energy pair
\(\ket{j,n},\ket{0,n+1}\), \(0\le n<4\), by squared
sine \(t\in[0,1]\), and act identically on the remaining
basis states. This is an energy-preserving unitary.
The net population transferred downward is
\[
A(q)t(p_j-qp_0).
\]
With \(s=\sqrt{1-t}\), the multipliers for a low-endpoint
spectator coherence, a high-endpoint spectator coherence, and
the internal coherence are, respectively,
\begin{align}
\ell(q,t)&=1-qA(q)(1-s),\\
\eta(q,t)&=1-A(q)(1-s),\\
\zeta(q,t)&=(w_0+w_4)s+(1-w_0-w_4)s^2 .
\label{eq:fiveLevelMultipliers}
\end{align}
These expressions follow by summing the bath-preserving
amplitudes. Bath-changing Kraus operators are single jumps and
do not contribute to off-diagonal output entries.

For the present example take \(q=e^{-1}\), and define
\begin{align}
d_x&=A(q)x\left(\frac15-\frac{7q}{20}\right),\\
t_a(x)&=\frac{1/10-d_x}
{A(q)[1/4-q(7/20+d_x)]},&
t_3&=\frac1{A(q^3)(4-9q^3)},\\
\widehat t_3&=\frac1{A(q^3)(4-7q^3)},&
\widehat t_a&=\frac2{A(q)(5-8q)}.
\label{eq:mixedStrengths}
\end{align}
Let \(W_x\) consist of the ordered collisions
\((0,1b,x)\), \((0,1a,t_a(x))\), \((0,3,t_3)\).
Let \(\widehat W\) use \((0,3,\widehat t_3)\) followed
by \((0,1a,\widehat t_a)\).

\begin{proposition}[Thermal lower bounds]
\label{prop:mixedWitnesses}
The channels \(W_{7/8},\widehat W,W_{1/7}\) are finite-bath
thermal operations with the target occupations
\((0.50,0.35,0.15)\). Their separately selected scores satisfy
\begin{equation}
L_1(W_{7/8})+L_2(\widehat W)+L_3(W_{1/7})
>0.09634328.
\label{eq:mixedTOLower}
\end{equation}
Their bath dimensions are \(125,25,125\), respectively.
\end{proposition}
\begin{proof}
The fixed interval checks give \(0<t<1\) for all strengths
in Eq.~\eqref{eq:mixedStrengths}. The successive populations
of \(W_x\) are
\begin{align*}
(7/20,1/4,1/5,1/5)
&\longmapsto(7/20+d_x,1/4,1/5-d_x,1/5)\\
&\longmapsto(9/20,3/20+d_x,1/5-d_x,1/5)\\
&\longmapsto(1/2,3/20+d_x,1/5-d_x,3/20).
\end{align*}
The middle occupations sum to \(7/20\).
For \(\widehat W\) the first transfer is \(1/20\), the
second is \(1/10\), and the final vector is
\((1/2,3/20,1/5,3/20)\).
These are exact identities for \(q=e^{-1}\).

Equation~\eqref{eq:fiveLevelMultipliers} gives the selected
scores
\begin{align}
T_1&=\frac1{20}\ell(q,7/8)\zeta(q,t_a(7/8))\ell(q^3,t_3),\\
T_2&=\frac1{20}\eta(q^3,\widehat t_3),\\
T_3&=\frac1{25}\ell(q,1/7)\ell(q,t_a(1/7))\zeta(q^3,t_3).
\end{align}
All are positive. Enclosing \(e^{-1}\) by consecutive alternating
Taylor sums and the square roots by outward integer bounds
gives \(T_1+T_2+T_3>0.09634328\).
The complete rational-interval check is in the Supplemental Material~\cite{SM}.
Composition multiplies the five-level bath dimensions,
giving the stated sizes.
\end{proof}

For this input, each selected output coherence depends
only on the corresponding input entry. Hence the maximum
CGP sum is the maximum, over population matrices satisfying
the prescribed constraints, of
\[
0.05\sqrt{P_{00}P_{1a,1a}}
+0.05\sqrt{P_{1b,1b}P_{33}}
+0.04\sqrt{P_{00}P_{33}}.
\]
At each admissible \(P\), the channel from
Proposition~\ref{prop:fixedP} attains all three terms.
This gives a population formulation of the CGP optimization
used in the Choi certificate. A thermal realization of a
fixed-\(P\) optimum would still require a separate construction.

\section{Phase readout and quantum Fisher information}
\label{app:phase}

For the family in Eq.~\eqref{eq:phaseFamily}, the phases occur in two separate \(2\times2\) blocks.
We apply one fixed time-translation covariant channel to
the whole family, independent of the unknown
\(\theta_{\mathrm e},\theta_{\mathrm o}\).
When the nonzero Bohr modes are one dimensional, each
selected coherence is multiplied by its own
parameter-independent coefficient, and the output populations
do not depend on either phase.

Let us consider one output block, denoting its off-diagonal
entry by \(z_\ell\) and its total population by \(s_\ell\).
Conditioned on occupying this block, a balanced two-path
measurement with readout phase \(\varphi\) has probability
\[
p_+(\varphi)
=\frac12+\frac{\operatorname{Re}(e^{i\varphi}z_\ell)}{s_\ell}.
\]
Its fringe visibility is therefore
\begin{equation}
V_\ell=\frac{2|z_\ell|}{s_\ell}.
\label{eq:visibility}
\end{equation}

\begin{proposition}[Phase Fisher information]
\label{prop:phaseInformation}
Let the Hamiltonian be nondegenerate with one-dimensional
nonzero Bohr modes, and let \(\Phi\) be a fixed
time-translation covariant channel, independent of the
parameters in Eq.~\eqref{eq:phaseFamily}.
For output blocks of positive weight \(s_\ell\), the quantum
Fisher matrix is
\begin{equation}
\mathcal F_{\ell\ell}=\frac{4|z_\ell|^2}{s_\ell},
\qquad
\mathcal F_{\mathrm e\mathrm o}=0.
\label{eq:phaseFisher}
\end{equation}
A block of zero weight contributes zero Fisher information.
A block-resolved measurement attains this matrix locally
at a calibrated phase point.
\end{proposition}
\begin{proof}
On one block,
\(\partial_{\theta_\ell}\sigma=
\begin{psmallmatrix}0&-iz_\ell\\iz_\ell^*&0\end{psmallmatrix}\).
The symmetric logarithmic derivative is
\[
L_\ell=\frac2{s_\ell}
\begin{pmatrix}0&-iz_\ell\\iz_\ell^*&0\end{pmatrix},
\]
extended by zero on the other block.
Multiplying the two block matrices verifies
\((\sigma L_\ell+L_\ell\sigma)/2
=\partial_{\theta_\ell}\sigma\). Moreover,
\(L_\ell^2=(4|z_\ell|^2/s_\ell^2)I\) on the block.
Taking \(\Tr(\sigma L_\ell^2)\) therefore gives the diagonal
entries in Eq.~\eqref{eq:phaseFisher}.
The cross term vanishes because the two derivatives have
disjoint support.

To attain both entries locally, we first resolve the two
blocks and then measure the eigenbasis of the corresponding
\(L_\ell\). This measurement is chosen at the calibrated
phase point. Its Fisher information is verified explicitly
in Supplemental Material~\cite{SM}.
\end{proof}

Both the analytic target and Eq.~\eqref{eq:finiteTask}
have \(s_{\mathrm e}=s_{\mathrm o}=1/2\).
For the reference finite-temperature input, this gives
\(V_\ell=4c r_\ell\) and
\(\mathcal F_{\ell\ell}=8c^2r_\ell^2\).
Thus the balanced coherence score corresponds to two
simultaneous visibility guarantees. Relative to the
individual phase-information optima, its worst Fisher
fraction is \(\alpha^2\).
Equation~\eqref{eq:nearOptimal} implies that the finite
balanced construction attains more than \(99.64\%\) of
the optimal worst normalized Fisher information.

Under the basis identification
\(0=00,1=01,2=10,3=11\), the reference Hamiltonian is
\[
H=3n_A+n_B+3n_An_B,
\qquad
n_A=\ee11\otimes I_2,\quad
n_B=I_2\otimes\ee11 .
\]
The two coherent blocks are the even- and odd-parity
sectors of an interacting two-qubit register.
The two readouts act on different blocks and can be performed
together. The restriction concerns how well their coherences
can be preserved during the same thermal population change.
Coherent control may be used to prepare the input and to
measure the phases, but no additional coherent resource is
supplied during the thermal operation.

\section{Proofs and finite baths at infinite temperature}
\label{app:infinite}

We consider the task in Eqs.~\eqref{eq:analyticPopulations} and~\eqref{eq:analyticInput}.
The two \(2\times2\) blocks have determinants
\(a(a+3h)-c^2>0\) and
\(a(a+3h)+2h^2-c^2>0\). Hence the state is positive
definite and normalized. Every target population is positive.

At \(\beta=0\), the admissible population matrices obey
\[
P\ge0,\qquad P\one=\one,\qquad
\one^TP=\one^T,\qquad Pn=n+xv_0,
\]
where \(n=(3,2,1,0)^T\) and \(v_0=(-1,-1,1,1)^T\).

\begin{proof}[Proof of Theorem~\ref{thm:analyticRegion}]
Put \(s_i=P_{ii}\). The row equations imply
\[
s_0,s_3\le1-x/3,\qquad s_1,s_2\le1-x/2.
\]
For example,
\(x=\sum_{j\ne0}P_{0j}(3-n_j)\le3(1-s_0)\);
the other rows give the remaining bounds.

For the coupled inequality take
\[
z=\tfrac14(0,-1,0,3)^T,\quad
u_0=\tfrac14(-3,-3,3,3)^T,\quad
t=\tfrac14(3,3,-1,-1)^T,\quad
b=\tfrac12(3,1,1,3)^T .
\]
The entrywise slacks satisfy
\[
z_i+(u_0)_r+t_rn_i-b_i\delta_{ri}
=
\begin{pmatrix}
0&1/2&0&0\\
3/2&0&0&0\\
0&0&0&3/2\\
0&0&1/2&0
\end{pmatrix}_{ri}\ge0.
\]
Multiplication by \(P_{ri}\) and summation give
\[
\frac{3s_0+s_1+s_2+3s_3}{2}\le4-2x.
\]
The scalar retention bounds and
\(\sqrt{ab}\le(a+b)/2\) yield all three inequalities in
Eq.~\eqref{eq:analyticRegion}.

Define the two feasible population matrices
\begin{equation}
P_A(x)=
\begin{pmatrix}
1-x/2&0&x/2&0\\
0&1-x/2&0&x/2\\
x/2&0&1-x/2&0\\
0&x/2&0&1-x/2
\end{pmatrix},\quad
P_B(x)=
\begin{pmatrix}
1-x/3&0&0&x/3\\
0&1-x&x&0\\
0&x&1-x&0\\
x/3&0&0&1-x/3
\end{pmatrix}.
\label{eq:analyticRoutes}
\end{equation}
Each is doubly stochastic and maps \(n\) to \(n+xv_0\).
Their fixed-\(P\) CGP channels attain \(A_x\) and \(B_x\).
For \(P_\lambda=(1-\lambda)P_A+\lambda P_B\), the
retentions are
\[
1-x/2+\lambda x/6,\qquad 1-x/2-\lambda x/2.
\]
They cover the segment, and independent attenuation gives its
downward closure, which is the region in
Eq.~\eqref{eq:analyticRegion}.

To obtain thermal implementations, use the uniform finite
ladder described below. On pairs \(02,13\), symmetric
transition \(x/2\) in each pair gives exactly \(P_A(x)\).
On pairs \(03,12\), transitions \(x/3,x\) give exactly
\(P_B(x)\). Their coherence deficits tend to zero with the
ladder size. Rational mixtures of the finite-bath channels
preserve the target exactly and approach any point of the
segment. Attenuation then approaches any point below it.
Compactness of the channel space gives the corresponding
closure channels. A point strictly below all upper facets is
strictly dominated by a segment point; sufficiently accurate
finite-bath approximations still dominate it, and attenuation
sets the two coordinates exactly. The full construction is in
the Supplemental Material~\cite{SM}.
\end{proof}

\subsection{Finite-bath realization}

For a pair of gap \(\Delta>0\), append a uniform Gibbs ladder
\(0,\Delta,\ldots,m\Delta\). To implement a symmetric transition
\(d\le m/(m+1)\), rotate every pair
\(\ket{j,n},\ket{i,n+1}\), \(0\le n<m\), with cosine
\[
s_m(d)=\sqrt{1-d(m+1)/m}.
\]
Both jump probabilities are exactly \(d\). The multipliers for
an endpoint-to-spectator coherence and for the internal
coherence are, respectively,
\begin{equation}
a_m(d)=\frac{1+ms_m(d)}{m+1},\qquad
b_m(d)=\frac{2s_m(d)+(m-1)s_m(d)^2}{m+1}.
\label{eq:uniformLadder}
\end{equation}
The states \(\ket{i,0}\) and \(\ket{j,m}\) have stay
amplitude one. All other states with system label \(i\)
or \(j\) have stay amplitude \(s_m(d)\).
Their weighted sums give
Eq.~\eqref{eq:uniformLadder}, and
\begin{align}
(1-d)-a_m(d)^2
&=\frac{m[1-s_m(d)]^2}{(m+1)^2},\\
(1-d)-b_m(d)
&=\frac{[1-s_m(d)]^2}{m+1}.
\end{align}
Each deficit is at most \(1/(m+1)\).

For the two routes in Eq.~\eqref{eq:analyticRoutes}, the
selected retentions are
\[
\big(a_m(x/2)^2,a_m(x/2)^2\big),\qquad
\big(b_m(x/3),b_m(x)\big).
\]
For every admissible finite \(m\), the transition matrices
are exactly \(P_A(x)\) and \(P_B(x)\); only the retained
coherences need to be approximated. For \(0<x\le1/2\),
we can approximate any point of
Eq.~\eqref{eq:analyticRegion} to coordinate error
\(\varepsilon\in(0,1]\) with bath dimension at most
\begin{equation}
4\left\lceil\frac2\varepsilon\right\rceil^3.
\label{eq:analyticBathSize}
\end{equation}
To obtain this estimate, we use two ladders for the endpoint
constructions, a uniform Gibbs flag for a rational approximation
of the mixing weight, and two degenerate qubits for independent
attenuation. The choice of their dimensions, and the extension
to \(1/2<x<1\), are given in Supplemental Material~\cite{SM}.

\subsection{Maximum sum and balanced protection}

The individual normalized maxima are
\[
m_{\mathrm e}=1-x/3,\qquad m_{\mathrm o}=1-x/2.
\]
For positive weights the support is
\begin{equation}
\max\left\{
(w_{\mathrm e}+w_{\mathrm o})(1-x/2),\
w_{\mathrm e}(1-x/3)+w_{\mathrm o}(1-x)
\right\}.
\label{eq:analyticSupport}
\end{equation}
Along the segment, the objective has slope proportional to
\(w_{\mathrm e}-3w_{\mathrm o}\). Hence the maximizing
endpoint changes at \(w_{\mathrm e}=3w_{\mathrm o}\);
in particular, \(A_x\) maximizes the unweighted sum.

\begin{proof}[Proof of Corollary~\ref{cor:analyticCosts}]
The individual and summed values follow from
Eqs.~\eqref{eq:analyticRegion} and~\eqref{eq:analyticSupport}.
If \(r_{\mathrm e}/m_{\mathrm e}\ge\alpha\) and
\(r_{\mathrm o}/m_{\mathrm o}\ge\alpha\), then
\[
\alpha(3m_{\mathrm e}+m_{\mathrm o})
\le3r_{\mathrm e}+r_{\mathrm o}\le4-2x.
\]
This gives Eq.~\eqref{eq:analyticAlpha}.
The resulting pair lies on the attained segment and is the
only pair with both coordinates at least those values.
Substitution gives Eq.~\eqref{eq:analyticBalanceCost}.
The exact-target approximations in
Theorem~\ref{thm:analyticRegion} transfer the continuous
optima to finite-TO suprema.
\end{proof}

\section{Positive-temperature benchmarks and the covariant boundary}
\label{app:positive}

All quantities in this appendix refer to Eqs.~\eqref{eq:finiteHamiltonian} and~\eqref{eq:finiteTask}.

\subsection{Population routes and exact constants}

For a pair \(i<j\), let \(q_{ij}=q^{E_j-E_i}\) and
\[
d_{ij}=\frac{1/50}{q_{ij}p_i-p_j}.
\]
A two-level population block with downward probability \(d_{ij}\)
and upward probability \(q_{ij}d_{ij}\) transfers \(1/50\)
upward for this input. The pairs \(02,13\) define a feasible
route \(P_A\); the pairs \(03,12\) define another route \(P_B\).
The resulting stay probabilities are
\begin{align}
\diag P_A&=
\left(\frac{569}{669},\
\frac{1028439}{1228439},\
\frac{5359}{6690},\
\frac{4370634}{6142195}\right),\\
\diag P_B&=
\left(\frac{18512829}{20512829},\
\frac{19}{29},\
\frac{169}{290},\
\frac{83076974}{102564145}\right).
\label{eq:finiteStay}
\end{align}
Their off-diagonal entries are the two pair transitions just
specified, with all other jumps zero.

Define positive constants by
\begin{equation}
\begin{aligned}
u^2&=\frac{828963582}{1369709485},&
v^2&=\frac{1837134867}{2739418970},\\
e^2&=\frac{1537989813499446}{2103880767916205},&
f^2&=\frac{3211}{8410}.
\end{aligned}
\label{eq:finiteConstants}
\end{equation}
Thus the fixed-\(P\) CGP retentions of \(P_A,P_B\) are
\((u,v)\) and \((e,f)\). Numerically,
\[
u=0.777953240551\ldots,\qquad
v=0.818919646779\ldots,\qquad
e=0.855000104272\ldots.
\]

\subsection{The individual thermal suprema and the maximum sum}

\begin{proof}[Proof of Theorem~\ref{thm:finiteExact}]
For the upper bounds, use the population reduction
of Eq.~\eqref{eq:populationReduction}.
The ratios \(p_i/g_i\) strictly decrease in \(i\).
The Gibbs and target row equations for levels 0 and 3
give
\[
\frac1{50}
\le(p_0/g_0-p_3/g_3)g_0(1-P_{00}),\qquad
\frac1{50}
\le(p_0/g_0-p_3/g_3)g_3(1-P_{33}).
\]
They bound the two stay probabilities by their values
in \(P_B\), and hence \(r_{\mathrm e}\le e\).
For the odd and summed bounds, rational population
certificates together with
\[
\sqrt{ab}\le\frac{\lambda a+\lambda^{-1}b}{2},
\qquad \lambda>0,
\]
give \(r_{\mathrm o}\le v\) and
\(r_{\mathrm e}+r_{\mathrm o}\le u+v\).
For the summed certificate, the coefficients are linear
combinations of \(u,v\) with rational coefficients.
The explicit dual table, every slack comparison, and the
nonzero determinant giving uniqueness of \(P_A\) are in
Supplemental Material~\cite{SM}.

The construction below approaches \((u,v)\) with exactly
\(P_A\). Two independent geometric-ladder constructions
on \(03,12\) approach \((e,f)\) with exactly \(P_B\).
They give the matching thermal lower bounds.
The positive slack coordinates of the summed dual restrict
an equality case to \(P_A\). At that matrix the individual
upper bounds are \(u,v\), so summed equality forces both.
Any nonconvergent selected pair in a maximizing finite-TO
sequence would have a CGP accumulation point contradicting
this uniqueness. Finally the gap is
\((e+v-u-v)/20\).
\end{proof}

\subsection{Geometric Gibbs ladders}

Optimal two-level coherence retention is known
\cite{PhysRevLett.115.210403,HuDing2019}.
We use a finite-ladder construction with exact transition
probabilities, given in Supplemental Material~\cite{SM}.
For a pair with gap \(\Delta>0\) and a prescribed downward
probability \(d\in(0,1)\), put
\(q_\Delta=e^{-\beta\Delta}\) and
\(Z_k=\sum_{n=0}^k q_\Delta^n\).
Choose \(m\ge1\) so that \(d<Z_{m-1}/Z_m\).
Such a choice is possible because \(Z_{m-1}/Z_m\to1\).
The bath has levels \(0,\Delta,\ldots,m\Delta\), and its
rotation cosines are \(y_m^{(n+1)/2}\), where \(y_m\in(0,1)\)
solves
\[
\frac{1}{Z_m}\sum_{n=0}^{m-1}q_\Delta^n(1-y_m^{n+1})=d.
\]
The left-hand side is continuous and strictly decreasing
in \(y_m\), from \(Z_{m-1}/Z_m\) to zero, so this solution
exists and is unique.
The downward and upward probabilities are exactly \(d\)
and \(q_\Delta d\). The internal coherence deficit
is bounded by
\[
0\le\sqrt{(1-d)(1-q_\Delta d)}-\kappa_m
\le
\sqrt{\frac{1-q_\Delta d}{1-d}}\frac{q_\Delta^m}{Z_m}.
\]
This gives the route \(P_B\). The cross-paired route
\(P_A\) needs coordinated amplitudes, as follows.

\subsection{A shared Gibbs bath for the sum-optimal route}

Write \(A_i=(P_A)_{ii}\), let \(Q=q^3\), and put
\[
L=\frac{A_2}{A_1},\qquad M=\frac{A_3}{A_0}.
\]
Both lie in \((0,1)\). The gaps of the pairs \(02\) and
\(13\) are 3 and 6, so one bath with levels \(3n\)
supports both. Define cosine sequences
\begin{equation}
u_n=
\begin{cases}
\sqrt L,&n=0,1,\\
\sqrt L\,v_{n-2},&n\ge2,
\end{cases}
\qquad
v_n=
\begin{cases}
\sqrt M,&n=0,\\
\sqrt M\,u_{n-1},&n\ge1.
\end{cases}
\label{eq:sharedRecursion}
\end{equation}
Use \(u_n\) to rotate \(\ket{2,n}\) with \(\ket{0,n+1}\),
and \(v_n\) to rotate \(\ket{3,n}\) with \(\ket{1,n+2}\).
The two families have disjoint system labels, so their
rotations form one energy-preserving unitary.

For weights \((1-Q)Q^n\), the squared-amplitude sums
\(T_2,T_3\) solve
\[
T_2=L(1-Q^2+Q^2T_3),\qquad
T_3=M(1-Q+QT_2).
\]
Their unique solution is \(T_2=A_2,T_3=A_3\), since
\[
A_0=1-Q+QA_2,\qquad A_1=1-Q^2+Q^2A_3.
\]
The recursion in Eq.~\eqref{eq:sharedRecursion} makes the
corresponding stay-amplitude vectors proportional, giving
\(v=\sqrt L A_1\) and \(u=\sqrt M A_0\).
To obtain a finite-bath operation, we truncate to
\(0,\ldots,m\). Denote the resulting downward probabilities
by \(D_{2,m}\) and \(D_{3,m}\). We correct them by changing
the two squared cosines at \(n=0\) to
\begin{equation}
\widetilde u_0^{\,2}
=L-\frac{1-A_2-D_{2,m}}{w_0},\qquad
\widetilde v_0^{\,2}
=M-\frac{1-A_3-D_{3,m}}{w_0},
\label{eq:sharedCorrection}
\end{equation}
where \(w_0=1/\sum_{n=0}^mQ^n\).
As \(m\) increases, the corrections vanish, so both modified
squared cosines remain in \((0,1)\) for all sufficiently
large \(m\). They then restore the downward probabilities
\(1-A_2\) and \(1-A_3\), and hence the whole population
matrix \(P_A\), exactly. The retained coherences converge
to \(u,v\): the remaining errors are bounded by the geometric
tails estimated in Supplemental Material~\cite{SM}.

With temperature-dependent constants, the same constructions
and strict support inequalities hold on an open interval
of inverse temperatures containing \(\log(11/10)\).
We prove the existence of this interval in Supplemental
Material~\cite{SM}. All numerical values quoted above refer
to the reference temperature.

\subsection{The CGP boundary and a stricter thermal constraint}

Equation~\eqref{eq:finiteExact} concerns the two individual
objectives and their unweighted sum. It does not determine
the optimum for other weights, which we now examine.
The complete positive-temperature
CGP boundary is obtained from two affine population arcs
\begin{equation}
P^{(1)}(s)=(1-s)P_A+sP_C,\qquad
P^{(2)}(s)=(1-s)P_C+sP_B,\qquad 0\le s\le1.
\label{eq:positiveArcs}
\end{equation}
The rational matrix \(P_C\) is given in the Supplemental Material~\cite{SM}; its diagonal is
\[
\left(
\frac{569}{669},\
\frac{461072041}{594872041},\
\frac{2621276267719}{3979693954290},\
\frac{83076974}{102564145}
\right).
\]
Let
\[
X_\ell(s)=\sqrt{P^{(\ell)}_{00}(s)P^{(\ell)}_{33}(s)},\qquad
Y_\ell(s)=\sqrt{P^{(\ell)}_{11}(s)P^{(\ell)}_{22}(s)}.
\]
In Supplemental Material~\cite{SM}, we prove that the arcs
\((X_\ell(s),Y_\ell(s))\) give the entire CGP Pareto boundary
and that each boundary point fixes its population matrix
uniquely. The proof uses polynomial dual slacks whose signs
are certified for every \(s\in[0,1]\).

The unique balanced CGP point lies on the first arc. Its
parameter is the root \(s_\star\in(0,1)\) of
\begin{equation}
e^2P^{(1)}_{11}(s)P^{(1)}_{22}(s)
-v^2P^{(1)}_{00}(s)P^{(1)}_{33}(s)=0,
\label{eq:positiveCGPbalance}
\end{equation}
a quadratic with rational coefficients. It gives
\[
\alpha_{\CGP}^\star
=0.93947683077729995098\ldots .
\]

A further property of a common thermal unitary excludes
the non-endpoint CGP Pareto points.

\begin{proposition}[Zero-pattern constraint at paired saturation]
\label{prop:zeroReciprocity}
Let \(E_0<E_1<E_2<E_3\) and \(\beta>0\).
A channel in \(\cTO\) satisfying
\[
P_{01}=P_{10}=P_{23}=P_{32}=0,\qquad
0<P_{33}\le P_{00},\qquad 0<P_{22}\le P_{11},
\]
and
\[
|\Phi(\ee03)_{03}|=\sqrt{P_{00}P_{33}},\qquad
|\Phi(\ee12)_{12}|=\sqrt{P_{11}P_{22}}
\]
must obey
\[
P_{ij}=0\ \Longleftrightarrow\ P_{ji}=0,
\qquad i\in\{0,1\},\quad j\in\{2,3\}.
\]
\end{proposition}

The proof is given in Supplemental Material~\cite{SM}.
It starts from simultaneous saturation, which aligns the
diagonal bath blocks in the Gibbs-weighted norm. We then
use row and column unitarity to relate their singular-value
distributions. The strict positive-temperature Gibbs weights
rule out residual mass at zero singular values and force the
stated zero pattern. The estimates are independent of bath
dimension, so the argument also applies to closure-TO.
Here the conclusion concerns which transitions can vanish
under the stated hypotheses; it is not a detailed-balance
condition for general thermal operations.

On the first arc, away from \(P_A\), one has
\(P_{03}=0<P_{30}\). On the second arc, away from \(P_B\),
one has \(P_{31}=0<P_{13}\).
The unique population matrix and the two saturated DMP
bounds at each Pareto point therefore imply
Eq.~\eqref{eq:thermalCGPBoundary}.
This is proved in the Supplemental Material~\cite{SM}.
Thus the classes agree on the individual and maximum-sum
benchmarks of this task, but not on its intermediate
Pareto boundary or balanced optimum.

\section{Thermal moment bounds and the balanced construction}
\label{app:thermal}

To bound thermal performance away from DMP saturation, we
impose moment constraints on the blocks of a common
energy-preserving unitary. The resulting bounds hold
uniformly over bath dimensions.

\subsection{A useful form of the diagonal bath blocks}

Write
\[
U=\sum_{r,i}\ee ri\otimes V_{ri},\qquad
D_i=V_{ii},\qquad
\varphi(X)=\Tr(\tau_BX).
\]
Every \(V_{ri}\) is a contraction. Energy conservation gives
\([H_B,V_{ri}]=(E_i-E_r)V_{ri}\).
The populations and retained coefficients are
\begin{equation}
P_{ri}=\varphi(V_{ri}^\dagger V_{ri}),\qquad
\kappa_{ij}=\varphi(D_j^\dagger D_i).
\label{eq:thermalBlocks}
\end{equation}

\begin{theorem}[Simultaneous improvement of retained coefficients]
\label{thm:diagonalDominance}
For every finite-bath thermal unitary, there is another
energy-preserving unitary on the same system and bath with
the same population matrix, whose diagonal bath blocks
\(\widetilde D_i\) are positive and diagonal in one common
bath energy basis, and satisfy
\begin{equation}
\Tr(\tau_B\widetilde D_j\widetilde D_i)
\ge|\Tr(\tau_BD_j^\dagger D_i)|
\quad\text{for all }i,j.
\label{eq:diagonalDominance}
\end{equation}
For simple nonzero modes, nonnegative weighted retention
objectives can therefore be upper-bounded in this form.
\end{theorem}
\begin{proof}
On each bath-energy eigenspace take
\(D_i=S_i\Sigma_iT_i^\dagger\), with \(\Sigma_i\)
nonnegative diagonal and its entries in decreasing order.
Choose one common coordinate basis for all \(\Sigma_i\).
The bath unitaries \(S_i,T_i\) commute with \(H_B\).
Set
\[
\widetilde U
=
\left(\sum_r\ee rr\otimes S_r^\dagger\right)
U
\left(\sum_i\ee ii\otimes T_i\right).
\]
Both factors conserve total energy, and
\(\widetilde V_{ri}=S_r^\dagger V_{ri}T_i\).
Its Gibbs-weighted squared norm is unchanged because
\(T_i\) commutes with \(\tau_B\).
Hence the population matrix is unchanged and
\(\widetilde D_i=\Sigma_i\).
The singular-value trace inequality, applied on each
bath-energy eigenspace where \(\tau_B\) is scalar, gives
Eq.~\eqref{eq:diagonalDominance} after summing the positive
Gibbs weights. Its finite-dimensional proof is supplied
in Supplemental Material~\cite{SM}.
\end{proof}

This constructs an improved dilation; it does not assert that
the original blocks commute. For a multidimensional mode
with several contributing input coherences,
Eq.~\eqref{eq:diagonalDominance} does not compare the
entire output block. We use it for the simple-mode task
in Eq.~\eqref{eq:finiteTask}.

\subsection{Thermal-moment semidefinite constraints}

For a bath operator \(A\) of energy shift \(\delta_A\),
\(\tau_BA=e^{-\beta\delta_A}A\tau_B\), so
\begin{equation}
\varphi(AB)=e^{-\beta\delta_A}\varphi(BA).
\label{eq:thermalCyclic}
\end{equation}
This Gibbs-weighted cyclic identity is imposed together
with row and column unitarity:
\begin{equation}
\sum_kV_{ki}^\dagger V_{kj}=\delta_{ij}I,\qquad
\sum_kV_{ik}V_{jk}^\dagger=\delta_{ij}I.
\label{eq:shellUnitaryBlocks}
\end{equation}
In the form of Theorem~\ref{thm:diagonalDominance},
\(0\le D_i\le I\) and \([D_i,D_j]=0\).

Take moments \(y_W=\varphi(W)\) of words in the blocks
\(V_{ri},V_{ri}^\dagger\), with \(y_\varnothing=1\).
Words of degree at most four suffice for the relaxation
used here. Finite phase averaging and complex conjugation
make the retained moments real, without changing the
population matrix or selected positive retentions.
Since every letter is a contraction, \(|y_W|\le1\).
Moment matrices and localizing matrices for the positive
diagonal blocks must be positive semidefinite.
Together with Eqs.~\eqref{eq:thermalCyclic}
and~\eqref{eq:shellUnitaryBlocks}, this defines a finite
outer SDP. The general noncommutative moment method is
described in Ref.~\cite{PironioNavascuesAcin2010};
the word list and thermal constraints used here are given
in Supplemental Material~\cite{SM}.

For this task there are 2942 real variables, 820 equality
constraints, and 84 positive-semidefinite blocks, the
largest of size 63. By
Theorem~\ref{thm:diagonalDominance}, a finite thermal
dilation can first be replaced by one with the same
population matrix and no smaller selected retentions.
After the phase averaging described above, its moments
are feasible for the relaxation. Thus the resulting upper
bound also bounds the original channel. Taking limits
of the bounded moments extends it to closure-TO, with
no restriction on the bath dimension.

We include the dual-stationarity residual explicitly in
the upper bound. Write the moment equalities as
\(\mathsf Ey=b\), the positive blocks as
\(\mathsf M_\ell(y)\ge0\), and the objective as \(h_w^Ty\).
For rational multipliers \(\nu\) and rational factors \(L_\ell\),
set
\begin{equation}
S_\ell=L_\ell L_\ell^T,\qquad
r=\mathsf E^T\nu-h_w-\sum_\ell\mathsf M_\ell^*(S_\ell).
\label{eq:momentResidual}
\end{equation}
Then, using \(y_0=1\) for the empty word,
\begin{equation}
h_w^Ty
=\nu^Tb-\sum_\ell\Tr[S_\ell\mathsf M_\ell(y)]-r^Ty
\le\nu^Tb-r_0+\sum_{j\ne0}|r_j|.
\label{eq:momentUpper}
\end{equation}
The right-hand side is rational and can be evaluated exactly.

\begin{table}[t]
\centering
\begin{tabular}{ccc}
\toprule
Weight \(w\) & Achievable lower bound on \(H_{\TO}(w,1)\)
& Thermal upper bound\\
\midrule
\(2\) & \(>2.37482612\) & \(<2.37532611\)\\
\(11/5\) & \(>2.53092875\) & \(<2.53418502\)\\
\(12/5\) & \(>2.69022350\) & \(<2.69614241\)\\
\bottomrule
\end{tabular}
\caption{\label{tab:thermalSupports}
Certified bounds for
\(H_{\mathfrak F}(w,1)=\sup_{\Gamma_{\mathfrak F}(\rho,p')}
(wr_{\mathrm e}+r_{\mathrm o})\).
Upper bounds apply to closure-TO with arbitrary baths.
Lower bounds come from explicit thermal constructions.
The numbers are normalized retentions; multiplication by
\(c=1/20\) gives absolute coherence units.}
\end{table}

\begin{theorem}[Intermediate thermal support bounds]
\label{thm:thermalSupports}
The upper bounds in Table~\ref{tab:thermalSupports}
hold for every channel in the exact-target closure-TO
feasible set. Each is strictly below achievable CGP
performance for the same weighted objective.
\end{theorem}
\begin{proof}
The three deposited rational Gram-factor certificates
satisfy Eq.~\eqref{eq:momentUpper} with right sides
strictly below the printed rational endpoints.
The verifier reconstructs the word identities, the equality
rows, and every coefficient residual exactly.
Supplemental Material~\cite{SM} lists the
exact rational upper bounds and their verification.
Feasible CGP channels give respective values greater than
\(2.37585720\), \(2.54053487\), and \(2.70653775\);
their rational population matrices and root enclosures
prove the strict comparisons.
\end{proof}

\subsection{Exact populations and a finite balanced protocol}

Proposition~\ref{prop:forest} gives convexity and downward
closure of the selected-coherence region. For this task,
we also prove density of exact-target finite TO in the
exact-target closure-TO set. Four thermal population
anchors surrounding \(p'\) correct the population error
of an approximating channel, with a mixing weight that
tends to zero. The construction is given in Supplemental
Material~\cite{SM}.

The same anchors give an exact finite implementation of
balanced protection. Let a core thermal process have output
populations \(r\) and nonnegative selected real coherences.
Put
\[
\epsilon_i=r_i-p'_i\quad(i=0,1,2),\qquad
S=\epsilon_0+\epsilon_1+\epsilon_2,\qquad
a_i=\epsilon_i-S/4,\quad a_3=-S/4,
\]
and let \(\epsilon=\max_{0\le j\le3}|a_j|\).
Choose \(h_0=10^{-4}\). The four anchor outputs are
\[
p'+h_0(\mathbf e_j-\mathbf e_3)\quad(j=0,1,2),
\qquad
p'-h_0\sum_{j=0}^2(\mathbf e_j-\mathbf e_3).
\]
Each is produced by four Gibbs-qubit collisions and has
strictly positive scalar coherence multipliers.
Mix the core with weight
\[
m=\frac{h_0}{h_0+4\epsilon}
\]
and the anchors with weights
\[
b_j=\frac{\epsilon-a_j}{h_0+4\epsilon}.
\]
The weights are nonnegative and sum to one. For \(i=0,1,2\),
\[
m\epsilon_i+h_0(b_i-b_3)=0,
\]
so the output populations are exactly \(p'\). The fourth
population follows from normalization. Each selected
coherence is at least \(m\) times the core coherence,
since the anchors contribute nonnegative coherences.

At \(\beta>0\), a mixture of baths \(H_j\) with partition
functions \(Z_j\) and weights \(\lambda_j>0\) can be
implemented by their direct sum, shifted by
\[
c_j=\beta^{-1}\log(Z_j/\lambda_j).
\]
Indeed, \(e^{-\beta c_j}Z_j=\lambda_j\).
The Gibbs state of the direct-sum bath has the desired
branch probabilities. The branch unitaries conserve total
energy, and their dimensions add. Four 16-dimensional
anchor baths therefore add at most 64 dimensions to the
core bath.

\begin{proof}[Proof of Theorem~\ref{thm:balancedThermal}]
For the core process we take
\(H_B=\bigoplus_{n=0}^{100}nI_2\), a bath of dimension 202.
The file \path{direct_balanced_protocol.json} specifies
1379 plane rotations within the total-energy shells
\(0,\ldots,107\). Each rotation is encoded by an integer
\(N\), from which we form the rational half-angle parameter
\(z=N/10^{15}\) and set
\[
\cos\theta=\frac{1-z^2}{1+z^2},\qquad
\sin\theta=\frac{2z}{1+z^2}.
\]
The identity \((1-z^2)^2+4z^2=(1+z^2)^2\) shows that each
rotation is exactly orthogonal. Since it acts within one
total-energy shell, it also conserves energy.

Taking the partial trace for each matrix unit with Gibbs
weights \(q^n/[2\sum_{k=0}^{100}q^k]\) gives rational
expressions for the populations and coherences. We evaluate
these expressions with outward integer intervals, obtaining
strictly positive lower bounds for both selected core
coherences. We then apply the population repair described
above, so that the final populations are exactly \(p'\).
The total anchor weight \(1-m\) is less than \(5.060881\times10^{-12}\),
and the conservative coherence lower bounds give
Eq.~\eqref{eq:balancedWitness}.
The complete repaired bath has dimension
\(202+4\cdot16=266\).
Supplemental Material~\cite{SM} gives the
exact certificate and the independently checkable data.
Dividing the two score lower bounds by upper root
enclosures of \(e,v\) proves
\(\alpha(\Phi_{\rm bal})>0.93693165\).

For any channel with balanced score \(\alpha\),
\[
\alpha\left(\frac{11}{5}e+v\right)
\le\frac{11}{5}r_{\mathrm e}+r_{\mathrm o}
\le U_{11/5},
\]
where \(U_{11/5}\) is the exact rational upper certificate
in Theorem~\ref{thm:thermalSupports}.
Root enclosures give
\(U_{11/5}/[(11/5)e+v]<0.938615\).
Exact-target density gives equality of the finite and closure
suprema. Finally,
\[
\frac{\alpha(\Phi_{\rm bal})}{\alpha^\star_{\TO}}
>
\frac{0.93693165}{0.938615}>0.9982.
\]
\end{proof}

\subsection{The loss in total retained coherence}

Maximizing the unweighted sum gives the pair \((u,v)\)
and worst normalized retention
\[
\frac ue=0.909886720088\ldots.
\]
By comparison, the balanced construction guarantees more
than \(0.93693165\) of each individual optimum at the same
target populations.

Every thermal channel
guaranteeing \(0.9369227\) of both individual optima
loses more than \(0.00119853\) in summed absolute coherence
relative to the sum supremum. Indeed, from the \(11/5\)
support bound and \(r_{\mathrm e}\ge0.9369227e\),
\[
\frac{u+v-r_{\mathrm e}-r_{\mathrm o}}{20}
\ge
\frac{u+v-U_{11/5}+(6/5)(0.9369227)e}{20}
>0.00119853.
\]
The last comparison is certified in the Supplemental Material~\cite{SM}. In particular, it applies
to the stronger guarantee of
Theorem~\ref{thm:balancedThermal}.

\section{Support minimization and population tolerance}
\label{app:robust}

\subsection{Balanced protection as a support problem}

For the compact convex selected region \(R\) in closure-TO,
define
\[
H_R(w_1,w_2)=\max_{r\in R}(w_1r_{\mathrm e}+w_2r_{\mathrm o}).
\]
The exact balanced value can be written as
\begin{equation}
\alpha_R
=
\min_{0\le t\le1}
H_R\left(\frac t e,\frac{1-t}{v}\right)
=
\min_{w>0}\frac{H_R(w,1)}{we+v}.
\label{eq:balancedSupport}
\end{equation}
The second equality applies here because the balanced optimum
is strictly less than one. To prove the first, normalize the
coordinates by their individual maxima. The largest point
on the diagonal of this compact downward-closed convex set
has a nonnegative supporting normal. Normalizing that
normal to sum to one gives \(t,1-t\).
Every other such normal bounds the diagonal point from
above, so minimizing the support gives equality.
The complete separation proof is shown in the Supplemental Material~\cite{SM}.

At most two points from one minimizing support face suffice
to form the balanced point in the closure.
The known endpoint constructions and Eq.~\eqref{eq:balancedInterval}
restrict any minimizing positive weight to
\[
2.0465754<w<2.8721918.
\]
The exact balanced value would follow from a matching
thermal support bound and a closure-TO construction on
the balanced ray.

\subsection{Population tolerance and output-state separation}

The moment multipliers also give an inequality without the
exact output-population constraint. For any thermal output
\(\sigma\) of the input in Eq.~\eqref{eq:finiteTask}, let
\(s=\diag\sigma\). Then
\begin{equation}
\frac{11}{5}r_{\mathrm e}+r_{\mathrm o}
\le U_{11/5}+\sum_{i=0}^3a_i(s_i-p'_i),
\label{eq:robustSupport}
\end{equation}
where
\[
a_0=\frac{577829406459}{40000000000},\quad
a_1=\frac{1644291531399}{80000000000},\quad
a_2=\frac{160602431958221}{100000000000000},\quad a_3=0 .
\]
Thus a tolerance \(\|s-p'\|_1\le\eta\) increases the
right side by at most \(10.27682207124375\,\eta\).
This follows by subtracting the midpoint of the largest
and smallest \(a_i\), using \(\sum_i(s_i-p'_i)=0\).
The certificate derivation is present in the Supplemental Material~\cite{SM}.

Let \(\sigma_\star\) be the balanced CGP output with
positive real coherences, defined by
Eq.~\eqref{eq:positiveCGPbalance}.
The same multipliers give a Hermitian witness combining
the two coherences with the output populations. Its
spectral bound proves
Eq.~\eqref{eq:stateSeparation}.
No final-population constraint is imposed on the thermal
competitor in Eq.~\eqref{eq:stateSeparation}.
The witness and all comparisons are given in the Supplemental Material~\cite{SM}.

\bibliographystyle{apsrev4-2}
\bibliography{references}
\end{document}